\documentclass[10pt,journal,compsoc]{IEEEtran}
\ifCLASSOPTIONcompsoc
  \usepackage[nocompress]{cite}
\else
  \usepackage{cite}
\fi
\ifCLASSINFOpdf
\else
\fi
\usepackage{amssymb}
\usepackage{amsmath}
\usepackage{amsthm}
\usepackage{array}
\usepackage{graphicx}
\usepackage{subfigure}
\usepackage{epstopdf}
\usepackage{cases}
\usepackage{color}
\usepackage{txfonts} % 默认字体times new roman
\usepackage{booktabs} %表格线加粗
\usepackage{float}
\usepackage{multirow}
\usepackage{picinpar}
\usepackage[colorlinks,
            linkcolor=red,
            anchorcolor=green,
            citecolor=blue
            ]{hyperref}
\usepackage{enumitem}
\usepackage[marginal]{footmisc}
\usepackage{fourier}
\usepackage{bookmark}

\begin{document}
\newtheorem{Theorem}{Theorem}[section]
\newtheorem{Lemma}{Lemma}[section]
\newtheorem{Corollary}{Corollary}[section]
\newtheorem{Remark}{Remark}[section]
\newtheorem{Proposition}{Proposition}[section]
\newtheorem{Definition}{Definition}[section]
%\newtheorem{Construction}{Construction}[section]
%\newtheorem{Condition}{Condition}[section]
%\newtheorem{Exa}[Theorem]{Example}
%\newcounter{claim_nb}[Theorem]
%\setcounter{claim_nb}{0}
%\newtheorem{claim}[claim_nb]{Claim}
\title{$\ell_{1}$-norm  quantile regression screening rule via the dual circumscribed sphere}

\author{Pan Shang,~\IEEEmembership{}
        Lingchen~Kong~\IEEEmembership{}
\IEEEcompsocitemizethanks{\IEEEcompsocthanksitem Pan Shang and Lingchen Kong are with the Department
of Applied Mathematics, Beijing Jiaotong University, Beijing, 100044.\protect\\
E-mail: 18118019@bjtu.edu.cn,konglchen@126.com}
\thanks{Manuscript received ; revised .}}


\markboth{ }%
{Shell \MakeLowercase{\textit{et al.}}: $\ell_{1}$-norm  quantile regression screening rule via the dual circumscribed sphere}



\IEEEtitleabstractindextext{%
\begin{abstract}
$\ell_{1}$-norm  quantile regression is a common choice if there exists outlier or heavy-tailed error in high-dimensional data sets. However, it is computationally expensive to solve this problem when the feature size of data is ultra high. As far as we know, existing screening rules can not speed up the computation of the  $\ell_{1}$-norm  quantile regression, which dues to the non-differentiability of the quantile function/pinball loss. In this paper, we introduce the dual circumscribed sphere technique and propose a novel $\ell_{1}$-norm  quantile regression screening rule. Our rule is expressed as the closed-form function of given data and eliminates inactive features with a low computational cost. Numerical experiments on some simulation and real data sets show that this screening rule can be used to eliminate almost all inactive features. Moreover, this rule can help to reduce up to 23 times of computational time, compared with the computation without our screening rule.
\end{abstract}

\begin{IEEEkeywords}
$\ell_{1}$-norm  quantile regression, Screening rule, Dual circumscribed sphere, Computational efficiency.
\end{IEEEkeywords}}

\maketitle

\IEEEdisplaynontitleabstractindextext

% \fi
\IEEEpeerreviewmaketitle

\IEEEraisesectionheading{\section{Introduction}\label{sec:introduction}}

\IEEEPARstart{T}{he} availability of high-dimensional data sets has been achieved with the help of advanced technologies.  In order to capture different characteristics in  these data sets, there are plenty of works about the regularized linear regression, composing of loss function and regularizer. The traditional least square method is sensitive to outliers and heavy-tailed errors in data, so there are many researches about the robust regression. One common choice is the $\ell_{1}$-norm quantile regression,  which is given as
\begin{equation*}
\underset{\beta}\min\sum\limits_{i=1}^{n}\rho_{\tau}(y_{i}-\textbf{x}^{T}_{i}\beta)+\lambda\|\beta\|_{1},
\end{equation*}
where $\rho_{\tau}(\cdot)$ is the quantile function/pinball loss and has been studied in many literatures.
See, e.g.,  Li and Zhu \cite{L08},  Bellon and Chernozhukoy \cite{B11},  Fan et al. \cite{F14}, Koenker \cite{K05}, Steinwart and  Christmann \cite{S07}\cite{S11}, Jumutc et al. \cite{JHS13}, Huang et al. \cite{HSS14}, Mkhadri et al. \cite{M17},  Yi and Huang \cite{Y17}, Gu et al. \cite{G18}.

For the $\ell_{1}$-norm quantile regression, there are many algorithms to solve it. For example, Koenker and Ng \cite{K05} introduced a modified version of the Frisch-Newton algorithm. Wu and Lange \cite{W08} propose the greedy coordinate descent algorithm for the LAD Lasso, a special case of the $\ell_{1}$-norm quantile regression.  Yi and Huang \cite{Y17} use the semismooth Newton coordinate descent algorithm to solve the elastic-net penalized quantile regression, involving the $\ell_{1}$-norm quantile regression.  Encouraged by the success of the alternating direction multiplier method (ADMM), Gu et al. \cite{G18} build the proximal ADMM  and a sparse coordinate descent ADMM to solve the quantile regression with different type regularizer, such as $\ell_{1}$ norm, adaptive $\ell_{1}$ norm and the folded concave regularizer. Koenker \cite{K15} set up an R package (quantreg) for solving the  quantile regression. There are no doubt that the standard matlab software for disciplined convex programming CVX (Micheal and Stephen \cite{M13}) can be used to solve the $\ell_{1}$-norm quantile regression. Nevertheless,  it is a challenging work to efficiently solve the $\ell_{1}$-norm quantile regression in ultra high-dimensional settings.

To speed up the computation of the regularized linear regression in ultra high-dimensional data sets, screening rules are proposed to eliminate inactive features. See, e.g., Fan et al. \cite{F08}, Ghaoui et al. \cite{G12}, Tibshirani et al. \cite{T12},  Wang et al. \cite{W15}, Wang et al.  \cite{W15b}, Ndiaye et al. \cite{N17}, Kuang et al. \cite{K17}, Xiang et al. \cite{X17}, Lee et al. \cite{L18}, Ren et al. \cite{R18}, Pan and Xu \cite{P19}. Roughly speaking, there are two types of screening rules: safe rule and heuristic rule. Safe rules mean all screened features are guaranteed to be inactive. For example, Ghaoui et al. \cite{G12} constructed SAFE rules to eliminate predictors for sparse supervised models, which includes Lasso, $\ell_{1}$-norm regularized logistic regression and $\ell_{1}$-norm regularized support vector machine. Wang et al. \cite{W15} built the dual polytope projection (DPP) and the enhanced version EDPP to discard inactive predictors for Lasso. Wang et al. \cite{W15b} analyzed the dual problem of the Fused Lasso and set up its screening rule via the monotonicity of the subdifferentials.  Ndiaye et al. \cite{N17}  built up  statics and dynamic  gap safe screening rules for Group Lasso, which are based on the gap between feasible points of Group Lasso and its dual problem. Ren et al. \cite{R18} proposed a novel bound propagation algorithm to screen out inactive features for general Lasso problems, including different differentiable loss functions. Instead, heuristic screening rules do not possess this advantage, but they can be more efficient by checking the Karush-Kuhn-Tucker (KKT) condition. For example, Tibshirani et al. \cite{T12} proposed strong rules for discarding inactive predictors under the assumption of the unit slope bound.  For robust loss functions, there are a few screening rules. For instance, Chen et al. \cite{C19} proposed the safe screening rules for the regularized Huber regression. To the best of our knowledge, existing screening rules mainly focus on regularized models with differentiable loss functions, such as the quadratic function, logistic function and Huber function. However, the quantile function/pinball loss is not differentiable, which means that the existing screening rules can not be directly applied to the  $\ell_{1}$-norm  quantile regression.

The core of an efficient screening rule is the easily expressed estimation of the dual solution, but the non-differentiability of the quantile function cause a difficulty to estimate the dual solution of the $\ell_{1}$-norm  quantile regression. In this paper, we introduce the dual circumscribed sphere technique to estimate the dual solution and build up a safe feature screening rule. In order to do so, we first present the dual problem of the $\ell_{1}$-norm  quantile regression, which maximizes a linear function under some linear constraints. With the help of duality theory,  we obtain a screening test for the $\ell_{1}$-norm  quantile regression based on its dual solution. This is the basic screening result that can be used to identity the inactive features, but it may be complex to get the dual solution if the data is high-dimensional. To make this screening test implementable, we estimate the dual solution via a region composing with two half spaces and some box constraints. In order to obtain the analytical form of the screening test,  we introduce the dual circumscribed sphere technique. With the help of this technique,  we develop an easily implementable screening rule for the $\ell_{1}$-norm  quantile regression, which has a closed-form formula of given data. As far as we know, this is the first time to use this technique in screening rules. According to this screening rule, we can eliminate the inactive features in data sets and reduce the computational cost of algorithms. For the purpose of illustrating the efficiency of our screening rule, we do some numerical experiments on simulation data and real data. The experiments present that the screening rule  can not only help to eliminate inactive features efficiently, but also save the computational time up to 23 times. Note that we adopt the proximal ADMM (Gu et al. \cite{G18}) to solve the $\ell_{1}$-norm  quantile regression. Nevertheless, our screening rule can be embedded into any algorithm or solver for solving this model.

The rest of this paper is organized as follows. We review basic concepts and results of the $\ell_{1}$-norm  quantile regression in Section 2.  In Section 3, we build up the $\ell_{1}$-norm  quantile regression screening rule via the dual circumscribed sphere technique.  In Section 4, we present the numerical results of some simulation data and real data.  Some conclusions are given in Section 5.

\section{Preliminaries}
Some basic concepts and results of the $\ell_{1}$-norm  quantile regression are reviewed in this section. See  Rockafellar \cite{R70} and Beck \cite{B17} for more details. First, we present the quantile function/pinball loss as follows. For any $\xi\in\mathbb{R}$,
\begin{equation}\label{eq:1}
\rho_{\tau}(\xi)=
\begin{cases}
\tau\xi,&\xi>0;\\
(\tau-1)\xi, &\xi\leq0.
\end{cases}
\end{equation}
Here $\tau\in(0,1)$ indicates the quantile of interest. Obviously, the quantile function is not differentiable at 0, so we introduce the concept of subdifferential.
\begin{Definition}
Let $f:\mathbb{R}^{n}\rightarrow (-\infty,+\infty]$ be a proper closed convex function and let $\textbf{x}\in \mathbb{R}^{n}$. A vector $\textbf{g}\in \mathbb{R}^{n}$ is called a subgradient of $f$ at $\textbf{x}$ if
\begin{center}
$f(\textbf{y})\geq f(\textbf{x})+\langle \textbf{g},\textbf{y}-\textbf{x}\rangle$,  $\forall \textbf{y}\in \mathbb{R}^{n}$.
\end{center}
The set of all subgradients of $f$ at $\textbf{x}$ is called the subdifferential of $f$ at $\textbf{x}$ and is denoted by $\partial f(\textbf{x})$, that is
\begin{center}
$\partial f(\textbf{x})=\left\{\textbf{g}\in \mathbb{R}^{n}: f(\textbf{y})\geq f(\textbf{x})+\langle \textbf{g},\textbf{y}-\textbf{x}\rangle, \forall \textbf{y}\in \mathbb{R}^{n} \right\}$.
\end{center}
\end{Definition}
According to this definition, we get the subdifferential of the quantile function.
\begin{Lemma}
Let $\rho_{\tau}: \mathbb{R}\rightarrow (-\infty,+\infty)$ be  the quantile function with $\tau\in(0,1)$. For any $\xi\in \mathbb{R}$, the subdifferential of $\rho_{\tau}$ at $\xi$ is
\begin{equation*}
\partial{\rho_{\tau}(\xi)}=
\begin{cases}
\{\tau\},&\xi>0;\\
\left[\tau-1,\tau\right],&\xi=0;\\
\{\tau-1\},&\xi<0.
\end{cases}
\end{equation*}
\end{Lemma}
For any $\textbf{\textit{x}}=(x_{1},x_{2},\cdots,x_{n})^{T}\in \mathbb{R}^{n}$, the $\ell_{1}$ norm of $\textbf{\textit{x}}$ is defined as $\|\textbf{\textit{x}}\|_{1}=\sum\limits_{i=1}^{n}|x_{i}|$. By simply calculation, the subdifferential of $\|\textbf{\textit{x}}\|_{1}$ is
$\partial\|\textbf{\textit{x}}\|_{1}=\left\{\left(\partial|x_{1}|,\cdots,\partial|x_{n}|\right)^{T}\right\}$,
 where
\begin{equation*}
\partial|x_{i}|=
\begin{cases}
\{1\},&x_{i}>0;\\
\left[-1,1\right],&x_{i}=0;\\
\{-1\},&x_{i}<0.
\end{cases}
\end{equation*}
Now, we review the definition of the conjugate function.
\begin{Definition}
Let $f:\mathbb{R}^{n}\rightarrow (-\infty,+\infty]$ be a proper closed convex function. The conjugate function of $f$ is denoted as $f^{*}$ and $f^{*}:\mathbb{R}^{n}\rightarrow (-\infty,+\infty]$  is defined as
\begin{center}
$f^{*}(\textbf{y})=\underset{\textbf{x}\in \mathbb{R}^{n}}\max\left\{\langle \textbf{y},\textbf{x}\rangle-f(\textbf{x})\right\}$, $\forall \textbf{y}\in \mathbb{R}^{n}.$
\end{center}
\end{Definition}
Based on this definition, we can get the conjugate of the quantile function.
\begin{Lemma}
Let $\rho_{\tau}: \mathbb{R}\rightarrow (-\infty,+\infty)$ be  the quantile function with $\tau\in(0,1)$. For any $\nu\in \mathbb{R}$, the conjugate of $\rho_{\tau}$ at $\nu$ is
\begin{equation*}
\rho^{*}_{\tau}(\nu)=\underset{\xi}\max\left\{\xi\nu-\rho_{\tau}(\xi)\right\}=
\begin{cases}
0,&\tau-1\leq\nu\leq\tau;\\
+\infty, &\rm{otherwise}.
\end{cases}
\end{equation*}
\end{Lemma}
In addition, the conjugate of $\ell_{1}$ norm is the indictor function of $\ell_{\infty}$ norm. That is, for any $\textbf{\textit{y}}\in \mathbb{R}^{n}$,
\begin{equation*}
\|\textbf{\textit{y}}\|^{*}_{1}=\underset{\textbf{\textit{x}}\in \mathbb{R}^{n}}\max\left\{\langle \textbf{\textit{y}},\textbf{\textit{x}}\rangle-\|\textbf{\textit{x}}\|_{1}\right\}=\delta_{\|\cdot\|_{\infty}\leq1}(\textbf{\textit{y}})=
\begin{cases}
0,&\|\textbf{\textit{y}}\|_{\infty}\leq1;\\
+\infty, &\rm{otherwise},
\end{cases}
\end{equation*}
where $\|\textbf{\textit{y}}\|_{\infty}$ is defined as
$\|\textbf{\textit{y}}\|_{\infty}=\max\{|y_{1}|,\cdots,|y_{n}|\}.$
Another norm that we use in this paper is $\ell_{2}$ norm. For any $\textbf{\textit{x}}=(x_{1},\cdots,x_{n})^{T}\in \mathbb{R}^{n}$, $\|\textbf{\textit{x}}\|_{2}=\sqrt{x^{2}_{1}+x^{2}_{2}+\cdots +x^{2}_{n}}$.

\begin{Definition}
Let $f:\mathbb{R}^{n}\rightarrow (-\infty,+\infty]$ be a proper closed convex function. The proximal mapping of $f$ is the operator given by
\begin{center}
prox$_{f}(\textbf{x})=\underset{\textbf{u}\in \mathbb{R}^{n}}{\arg\min}\left\{f(\textbf{u})+\frac{1}{2}\|\textbf{u}-\textbf{x}\|^{2}_{2}\right\},$ for any $\textbf{x}\in \mathbb{R}^{n}$.
\end{center}
\end{Definition}
Based on this definition, we review the proximal mapping of $\ell_{1}$ norm (Beck \cite{B17}) in the following lemma.
\begin{Lemma}
Denote $f=\lambda\|\textbf{x}\|_{1}:\mathbb{R}^{n}\rightarrow (-\infty,+\infty)$. For any $\textbf{x}\in \mathbb{R}^{n}$, the proximal mapping of $f$ is
\begin{center}
prox$_{f}(\textbf{x})=[|\textbf{x}|-\lambda e]_{+}\odot \textit{sign}(\textbf{x}),$
\end{center}
where $\odot$ denotes the Hadamard product and $\textit{sign}(\textit{\textbf{x}})=(\textit{sign}(\textit{x}_{1}),\textit{sign}(\textit{x}_{2}),\cdots,\textit{sign}(\textit{x}_{n}))^{T}$ with
$$\textit{sign}(\textit{x}_{i})=\begin{cases}
1,&x_{i}>0\\
0,&x_{i}=0\\
-1,&x_{i}<0
\end{cases}.$$
\end{Lemma}

By simple calculation, we get the proximal mapping of quantile function.
\begin{Lemma}
Let $\rho_{\tau}: \mathbb{R}\rightarrow (-\infty,+\infty)$ be  the quantile function with $\tau\in(0,1)$. For any $\xi\in \mathbb{R}$, the  proximal mapping of $\rho_{\tau}$ at $\xi$ is
\begin{align*}
\textit{prox}_{\rho_{\tau}}(\xi)&=\underset{\upsilon\in R}{\arg\min}\left\{\rho_{\tau}(\upsilon)+\frac{1}{2}(\xi-\upsilon)^{2}\right\}\\
&=
\begin{cases}
\xi-\tau,&\xi\geq\tau;\\
0,&\tau-1<\xi<\tau;\\
\xi-(\tau-1), &\xi\leq\tau-1.
\end{cases}
\end{align*}
\end{Lemma}
\section{$\ell_{1}$-norm  quantile regression screening rule}
In this section, we build up the $\ell_{1}$-norm  quantile regression screening rule. At first, we introduce the dual circumscribed sphere technique, which is used to estimate the dual solution. With the help of this technique, we obtain the safe feature screening rule in this paper.

The $\ell_{1}$-norm  quantile regression (Li and Zhu \cite{L08}) is
\begin{equation}\label{eq:2}
\underset{\beta}\min\sum\limits_{i=1}^{n}\rho_{\tau}(y_{i}-\textbf{\textit{x}}^{T}_{i}\beta)+\lambda\|\beta\|_{1},
\end{equation}
where $\beta=(\beta_{1},\cdots, \beta_{p})^{T}\in \mathbb{R}^{p}$ is the unknown coefficient vector and $\lambda>0$ is the tuning parameter. For ease of expression, we denote the response variable $\textbf{\textit{y}}=(y_{1},\cdots,y_{n})^{T}\in \mathbb{R}^{n}$  and the prediction matrix $X=\left(\textbf{\textit{x}}_{1},\textbf{\textit{x}}_{2},\cdots,\textbf{\textit{x}}_{n}\right)^{T}=(X_{\cdot1},\cdots,X_{\cdot p})\in \mathbb{R}^{n\times p}$, where $\textbf{\textit{x}}_{i}\in \mathbb{R}^{p}$ denotes the $i_{th}$ sample and $X_{\cdot j}\in \mathbb{R}^{n}$ the $j_{th}$ feature. In order to emphasize that the solution of (\ref{eq:2}) relies on the choice of $\lambda$, we denote $\beta^{*}(\lambda)$ as it. If $\beta^{*}_{j}(\lambda)=0$, the $j_{th}$ feature is uncorrelated to the $\tau_{th}$ quantile of $\textbf{\textit{y}}$ and we call it the inactive feature. Since that the quantile function is not differentiable, existing screening rules can not be directly applied to the $\ell_{1}$-norm  quantile regression  (\ref{eq:2}). To overcome the challenge of the non-differentiability of the quantile function, we propose the dual circumscribed technique, which is used to estimate the dual solution. In order to do so, we introduce a variable $\alpha=(\alpha_{1},\cdots,\alpha_{n})^{T}\in \mathbb{R}^{n}$ and transform the model  (\ref{eq:2}) to a constraint problem as below.
\begin{equation}\label{eq:3}
\begin{split}
&\underset{\beta,\alpha}\min\sum\limits_{i=1}^{n}\rho_{\tau}(\alpha_{i})+\lambda\|\beta\|_{1}\\
&s.t. \quad y_{i}-\textbf{\textit{x}}^{T}_{i}\beta-\alpha_{i}=0, i=1,\cdots,n.
\end{split}
\end{equation}
Similarly, the solution of (\ref{eq:3}) is denoted as $(\beta^{*}(\lambda),\alpha^{*}(\lambda))$, where $\beta^{*}(\lambda)$ is the solution of (\ref{eq:2}). By introducing the Lagrangian multiplier
$\theta=(\theta_{1},\cdots,\theta_{n})^{T}\in \mathbb{R}^{n},$
 we obtain the Lagrangian function of the model (\ref{eq:3}), which is
\begin{center}
$\textit{L}\left(\beta,\alpha,\theta\right)
=\sum\limits_{i=1}^{n}\rho_{\tau}(\alpha_{i})+\lambda\|\beta\|_{1}
+\sum\limits_{i=1}^{n}\theta_{i}(y_{i}-\textbf{\textit{x}}^{T}_{i}\beta-\alpha_{i}).$
\end{center}
By direct computation, we obtain that
\begin{align*}
&\underset{\beta,\alpha}\min\textit{L}\left(\beta,\alpha,\theta\right)\\
&=\underset{\beta}\min \left\{\lambda\|\beta\|_{1}-\langle X^{T}\theta,\beta\rangle\right\}
+\sum\limits_{i=1}^{n}\underset{\alpha_{i}}\min \left\{\rho_{\tau}(\alpha_{i})-\theta_{i}\alpha_{i}\right\}+\langle\theta,\textbf{\textit{y}}\rangle\\
&=-\delta_{\|\cdot\|_{\infty}\leq\lambda}\left(X^{T}\theta\right)
-\sum\limits_{i=1}^{n}\delta_{\tau-1\leq\theta_{i}\leq\tau}\left(\theta_{i}\right)+\langle\theta,\textbf{\textit{y}}\rangle,
\end{align*}
where the last equality is based on the result of Lemma 2.2 and the argument after it. Therefore, the Lagrangian dual form of the model (\ref{eq:2}) is $$\underset{\theta}\max\underset{\beta,\alpha}\min\textit{L}\left(\beta,\alpha,\theta\right).$$ That is,
\begin{equation}\label{eq:4}
\begin{split}
&\underset{\theta}\max \quad\langle\theta,\textbf{\textit{y}}\rangle\\
&s.t.\quad\left\|X^{T}\theta\right\|_{\infty}\leq \lambda,\\
&\quad\quad\quad \tau-1\leq\theta_{i}\leq\tau, \quad i=1,\cdots,n.
\end{split}
\end{equation}
Denote the  solution of the dual problem (\ref{eq:4}) as $\theta^{*}(\lambda)$. The Karush-Kuhn-Tucker (KKT) system of  (\ref{eq:3}) and  (\ref{eq:4}) is
\begin{eqnarray}\label{eq:5}
\begin{cases}
X^{T}\theta \in \lambda\partial\|\beta\|_{1},\\
y_{i}-\textbf{\textit{x}}^{T}_{i}\beta-\alpha_{i}=0,\\
\theta_{i}\in \partial{\rho_{\tau}}(\alpha_{i}), \quad i=1,2,\cdots,n.
\end{cases}
\end{eqnarray}
If a pair $\left(\beta(\lambda), \alpha(\lambda), \theta(\lambda)\right)$ satisfies the KKT system, it is called the KKT point of  (\ref{eq:3}) and  (\ref{eq:4}). Note that $(0,\textbf{y})$ is a feasible point of  the problem  (\ref{eq:3}). Slater constraint qualification (Rockafellar \cite{R70}) holds on this problem. Hence, we easily show the following duality theorem.
\begin{Theorem}\label{thm:1}
(\textbf{Strong duality theorem}) The optimal solutions of problems (\ref{eq:3}) and (\ref{eq:4}) compose a KKT point. Moreover, the optimal values of these two problems are same.
\end{Theorem}
According to Theorem \ref{thm:1}, for any $\lambda>0$, the solutions of problems (\ref{eq:3}) and (\ref{eq:4}) satisfies that
\begin{center}
$X^{T}\theta^{*}(\lambda)\in \lambda\partial\|\beta^{*}(\lambda)\|_{1}$.
\end{center}
This, together with the subdifferential of $\ell_{1}$ norm, leads to the basic idea of our screening test.
 \begin{Theorem}\label{thm:2}
Let $j\in\{1,\cdots,p\}$. For any tuning parameter $\lambda>0$, if the solution of (\ref{eq:4}) satisfies that
$$|X_{\cdot j}^{T}\theta^{*}(\lambda)|<\lambda,$$
then $\beta_{j}^{*}(\lambda)=0$, which means that $X_{\cdot j}$ is uncorrelated to the $\tau_{th}$ quantile of $\textbf{\textit{y}}$.
\end{Theorem}
In this paper, we call $|X_{\cdot j}^{T}\theta^{*}(\lambda)|<\lambda$ as the screening test. For any tuning parameter $\lambda>0$, according to this theorem, some inactive features  can be eliminated. Therefore, this screening test reduces the computational cost of solving the model (\ref{eq:2}) when the dual solution is known. However, the dual solution is unknown in most cases and  it may need high computational cost to solve the dual problem in high-dimensional settings.  So, we estimate the dual solution in the following part.

Before estimating the dual solution, we present the lower bound of the tuning parameter such that the solution of  (\ref{eq:2}) is zero.

Let $\theta_{max}=(\theta^{(1)}_{max},\cdots,\theta^{(n)}_{max})^{T}\in\mathcal{F}\subset\mathbb{R}^{n}$ with
\begin{equation}\label{eq:new}
\theta^{(i)}_{max}=\partial\rho_{\tau}(y_{i})=
\begin{cases}
\{\tau\},&y_{i}>0\\
\left[\tau-1,\tau\right],&y_{i}=0, \quad i=1,2,\cdots,n.\\
\{\tau-1\},&y_{i}<0
\end{cases}
\end{equation}
Denote
\begin{eqnarray}\label{eq:7}
I_{1}=\left\{i\Big|y_{i}>0,i=1,2,\cdots,n\right\}, I_{2}=\left\{i\Big|y_{i}<0,i=1,2,\cdots,n\right\}
\end{eqnarray}\
and $\Delta=(\Delta_{1},\Delta_{2},\cdots,\Delta_{p})^{T}\in\mathbb{R}^{p}$ with
\begin{equation}\label{eq:new2}
\Delta_{j}=\max~\begin{Bmatrix}
\zeta_{j}+\sum\limits_{i\notin I_{1}\cup I_{2},X_{ij}\geq0}\tau X_{ij}+\sum\limits_{i\notin I_{1}\cup I_{2},X_{ij}<0}(\tau-1)X_{ij},\\ -\zeta_{j}-\sum\limits_{i\notin I_{1}\cup I_{2},X_{ij}\geq0}(\tau-1)X_{ij}-\sum\limits_{i\notin I_{1}\cup I_{2},X_{ij}<0}\tau X_{ij}
\end{Bmatrix}
\end{equation}
where $\zeta_{j}=\tau\sum\limits_{i\in I_{1}}X_{ij}+(\tau-1)\sum\limits_{i\in I_{2}}X_{ij}$ and $j=1,2,\cdots,p$. We then obtain the next lemma.
\begin{Lemma} \label{lem:1}
Denote $\lambda_{max}=\|\Delta\|_{\infty}$. The following statements hold.\\
(i) If $\beta^{*}(\lambda)=0$, then $\lambda\geq\lambda_{max}$.\\
(ii) If $\lambda>\lambda_{max}$, then $\beta^{*}(\lambda)=0$.
\end{Lemma}
\begin{proof}
(i) Let $\beta^{*}(\lambda)=0$. By Theorem \ref{thm:1} and the KKT system (\ref{eq:5}), we know that $\textbf{\textit{y}}=\alpha$, $\theta^{*}(\lambda)\in\mathcal{F}$ and $\lambda\geq\|X^{T}\theta_{max}\|_{\infty}$ with any $\theta_{max}\in\mathcal{F}$ defined in (\ref{eq:new}).
Because $\mathcal{F}$ is a set, $\lambda$ needs to satisfy that $\lambda\geq\underset{\theta_{max}\in\mathcal{F}}\max\|X^{T}\theta_{max}\|_{\infty}$.
By the definition of the $\|\cdot\|_{\infty}$, we know that
\begin{align*}
\underset{\theta_{max}\in\mathcal{F}}\max\|X^{T}\theta_{max}\|_{\infty}&=\underset{\theta_{max}\in\mathcal{F}}\max~\underset{j}\max~|X_{.j}^{T}\theta_{max}|\\
&=\underset{j}\max~\underset{\theta_{max}\in\mathcal{F}}\max~|X_{.j}^{T}\theta_{max}|.
\end{align*}
Based on (\ref{eq:new}) and (\ref{eq:7}), we know that
\begin{align*}
&\underset{\theta_{max}\in\mathcal{F}}\max~|X_{.j}^{T}\theta_{max}|\\
&=\underset{\tau-1\leq\kappa_{i}\leq\tau}\max~\Big|\underset{\zeta_{j}}{\underbrace{\tau\sum\limits_{i\in I_{1}}X_{ij}+(\tau-1)\sum\limits_{i\in I_{2}}X_{ij}}}+\sum\limits_{i\notin I_{1}\cup I_{2}}\kappa_{i}X_{ij}\Big|\\
&=\underset{\tau-1\leq\kappa_{i}\leq\tau}\max~|\zeta_{j}+\sum\limits_{i\notin I_{1}\cup I_{2}}\kappa_{i}X_{ij}|\\
&=\max~\left\{\zeta_{j}+\sum\limits_{i\notin I_{1}\cup I_{2}}\underset{\tau-1\leq\kappa_{i}\leq\tau}\max\kappa_{i}X_{ij},
-\zeta_{j}+\sum\limits_{i\notin I_{1}\cup I_{2}}\underset{\tau-1\leq\kappa_{i}\leq\tau}\max-\kappa_{i}X_{ij}\right\}\\
&=\max~\begin{Bmatrix}
\zeta_{j}+\sum\limits_{i\notin I_{1}\cup I_{2},X_{ij}\geq0}\tau X_{ij}+\sum\limits_{i\notin I_{1}\cup I_{2},X_{ij}<0}(\tau-1)X_{ij},\\ -\zeta_{j}-\sum\limits_{i\notin I_{1}\cup I_{2},X_{ij}\geq0}(\tau-1)X_{ij}-\sum\limits_{i\notin I_{1}\cup I_{2},X_{ij}<0}\tau X_{ij}
\end{Bmatrix}\\
&=\Delta_{j}.
\end{align*}
Therefore,
\begin{align*}
\lambda\geq\underset{\mathcal{F}}\max~\|X^{T}\theta_{max}\|_{\infty}=\underset{j}\max~\Delta_{j}=\|\Delta\|_{\infty}=\lambda_{max}.
\end{align*}

(ii) Let  $\lambda>\lambda_{max}$. Denote $f_{\lambda}(\beta^{*}(\lambda))$ as the optimal value of the problem (\ref{eq:2}), $g_{\lambda}(\theta^{*}(\lambda))=\langle\theta^{*}(\lambda),\textbf{\textit{y}}\rangle$ as the objective function  of the problem (\ref{eq:4}) and
\begin{align*}
g&=\underset{\theta}\max\left\{\langle\theta,\textbf{\textit{y}}\rangle\Big|\tau-1\leq\theta_{i}\leq\tau, \quad i=1,\cdots,n\right\}\\
&=\tau\sum\limits_{i\in I_{1}}y_{i}+(\tau-1)\sum\limits_{i\in I_{2}}y_{i}.
\end{align*}
It is clear that $g\geq g_{\lambda}(\theta^{*}(\lambda))$, because there are more constraints of $\theta$ when calculating the value of $g_{\lambda}(\theta^{*}(\lambda))$. Based on $\lambda>\lambda_{max}$, we know that any $\theta_{max}\in\mathcal{F}$ is a feasible point of the problem (\ref{eq:4}) and $g_{\lambda}(\theta^{*}(\lambda))\geq \langle\theta_{max},\textbf{y}\rangle=g$. Therefore, for any $\lambda>\lambda_{max}$, $g= g_{\lambda}(\theta^{*}(\lambda))$ and any $\theta_{max}\in\mathcal{F}$ is a solution of the problem (\ref{eq:4}). Replacing $\theta_{max}$ into the KKT system (\ref{eq:5}), we obtain the results (a), (b) and (c).

(a) If $y_{i}>0$, we have $\theta^{(i)}_{max}=\{\tau\}$, which leads to $\alpha^{*}_{i}(\lambda)\geq0$ and $y_{i}-x^{T}_{i}\beta^{*}(\lambda)\geq0$. Therefore, $$\rho_{\tau}(y_{i}-x^{T}_{i}\beta^{*}(\lambda))=\tau(y_{i}-x^{T}_{i}\beta^{*}(\lambda)).$$

(b) If $y_{i}=0$, we have $\theta^{(i)}_{max}=\left[\tau-1,\tau\right]$, which leads to $\alpha^{*}_{i}(\lambda)=0$ and $y_{i}-x^{T}_{i}\beta^{*}(\lambda)=0$. Therefore, $$\rho_{\tau}(y_{i}-x^{T}_{i}\beta^{*}(\lambda))=0.$$

(c) If $y_{i}<0$, we have $\theta^{(i)}_{max}=\{\tau-1\}$, which leads to $\alpha^{*}_{i}(\lambda)\leq0$ and $y_{i}-x^{T}_{i}\beta^{*}(\lambda)\leq0$. Therefore, $$\rho_{\tau}(y_{i}-x^{T}_{i}\beta^{*}(\lambda))=(\tau-1)(y_{i}-x^{T}_{i}\beta^{*}(\lambda)).$$
According to these three results, we have
$$f_{\lambda}(\beta^{*}(\lambda))=\tau\sum\limits_{i\in I_{1}}(y_{i}-x^{T}_{i}\beta^{*}(\lambda))+(\tau-1)\sum\limits_{i\in I_{2}}(y_{i}-x^{T}_{i}\beta^{*}(\lambda)).$$
Based on the strong duality theorem, we know that
$$f_{\lambda}(\beta^{*}(\lambda))=g_{\lambda}(\theta_{max}),$$
which means
\begin{align*}
\tau\sum\limits_{i\in I_{1}}(y_{i}-x^{T}_{i}\beta^{*}(\lambda))+(\tau-1)\sum\limits_{i\in I_{2}}(y_{i}-x^{T}_{i}\beta^{*}(\lambda))\\
=\tau\sum\limits_{i\in I_{1}}y_{i}+(\tau-1)\sum\limits_{i\in I_{2}}y_{i}.
\end{align*}
Clearly, $\beta^{*}(\lambda)=0$ is a solution of (\ref{eq:2}) when $\lambda>\lambda_{max}$. Moreover, based on $X^{T}\theta_{max}\in \lambda\partial\|\beta\|_{1}$ and $\lambda>\lambda_{max}$, $\beta^{*}(\lambda)=0$ is the unique solution.
\end{proof}
\begin{Remark}
In Lemma 3.1, the expression of $\lambda_{max}$ seems complex, while it can be determined by the given data $X$ and $\textbf{\textit{y}}$. The complex expression of $\lambda_{max}$ is due to that $\mathcal{F}$ is a set. In the special case that all elements of $\textbf{\textit{y}}$ are not zero, $\mathcal{F}=\{\theta_{\max}\}$ is a singleton and $\lambda_{max}=\|X^{T}\theta_{max}\|_{\infty}$.
\end{Remark}
According to this result, we consider about $\lambda\in(0,\lambda_{max}]$ in the rest of this paper. Next, we estimate the dual solution when $\lambda\in(0,\lambda_{max}]$, via relaxing the dual box constraints.
\begin{Lemma}
For any $\lambda\in(0,\lambda_{max}]$, the solution of the problem (\ref{eq:4}) $\theta^{*}(\lambda)$ satisfies that
\begin{eqnarray}
\theta^{*}(\lambda)\in \Theta=\left\{\theta\Big|\langle\theta,\textbf{\textit{y}}\rangle\leq b_{1},\langle\theta,\textbf{\textit{y}}\rangle\geq b_{2},\|\theta\|_{2}\leq\rho\right\},
\end{eqnarray}
where $\rho=\frac{\sqrt{n}}{2},$
$$b_{1}=\tau\sum\limits_{i\in I_{1}}y_{i}+(\tau-1)\sum\limits_{i\in I_{2}}y_{i}-(\tau-\frac{1}{2})\sum\limits_{k=1}^{n}y_{k}$$
and
$$b_{2}=\frac{\lambda}{\lambda_{max}}\left(\tau\sum\limits_{i\in I_{1}}y_{i}+(\tau-1)\sum\limits_{i\in I_{2}}y_{i}\right)-(\tau-\frac{1}{2})\sum\limits_{k=1}^{n}y_{k}.$$
\end{Lemma}
\begin{proof}
For any tuning parameter $\lambda\in(0,\lambda_{max}]$, the dual solution $\theta^{*}(\lambda)$ satisfies that
\begin{center}
$\|X^{T}\theta^{*}(\lambda)\|_{\infty}\leq\lambda$ and $\tau-1\leq\theta^{*}_{i}(\lambda)\leq\tau, i=1,\cdots,n$.
\end{center}
It is sure that the $\hat{\theta}(\lambda)=\frac{\lambda}{\lambda_{max}}\theta_{max}$ satisfies these constraints for any $\theta_{max}\in\mathcal{F}$, so the dual solution must satisfy that
$$\langle\theta^{*}(\lambda),\textbf{\textit{y}}\rangle\geq g_{\lambda}(\hat{\theta}(\lambda))= \frac{\lambda}{\lambda_{max}}\left(\tau\sum\limits_{i\in I_{1}}y_{i}+(\tau-1)\sum\limits_{i\in I_{2}}y_{i}\right).$$ In addition,
$$\langle\theta^{*}(\lambda),\textbf{\textit{y}}\rangle\leq g=\tau\sum\limits_{i\in I_{1}}y_{i}+(\tau-1)\sum\limits_{i\in I_{2}}y_{i},$$
where $I_{1}$ and $I_{2}$ are defined in (\ref{eq:7}). This means
\begin{eqnarray}
\theta^{*}(\lambda)\in \left\{\theta \Big| \langle\theta,\textbf{y}\rangle\leq g, \langle\theta,\textbf{y}\rangle\geq g_{\lambda}(\hat{\theta}(\lambda)), \{\tau-1\leq\theta_{i}\leq\tau\}_{i=1}^{n}\right\}.
\end{eqnarray}
To simplify the computation of $|X_{\cdot j}^{T}\theta^{*}(\lambda)|$ in the screen test, let
\begin{eqnarray}
\mathcal{\Theta}_{1}=\left\{\theta \Big| \langle\theta,\textbf{y}\rangle\leq g,\langle\theta,\textbf{y}\rangle\geq g_{\lambda}(\hat{\theta}(\lambda)), \{\tau-1\leq\theta_{i}\leq\tau\}_{i=1}^{n}\right\}.
\end{eqnarray}
The dual solution $\theta^{*}(\lambda)$ must be in this set.

Introducing
$\tilde{\theta}\in \mathbb{R}^{n}$ with $\tilde{\theta}_{i}=\theta_{i}-(\tau-\frac{1}{2})$,
then $\mathcal{\Theta}_{1}$ can be expressed as
\begin{eqnarray}
\mathcal{\Theta}_{2}=\left\{\tilde{\theta}\Big| \langle\tilde{\theta},\textbf{y}\rangle\leq b_{1}, \langle\tilde{\theta},\textbf{y}\rangle\geq b_{2}, |\tilde{\theta}_{i}|\leq\frac{1}{2}, i=1,\cdots,n\right\},
\end{eqnarray}
where $b_{1}=g-(\tau-\frac{1}{2})\sum\limits_{i=1}^{n}y_{i}$ and $b_{2}=g_{\lambda}(\hat{\theta})-(\tau-\frac{1}{2})\sum\limits_{i=1}^{n}y_{i}$.
To get an easily calculated result of $\underset{\theta}\max|X_{.j}^{T}\theta|$, we relax $\mathcal{\Theta}_{2}$ as
\begin{eqnarray*}
\Theta=\left\{\tilde{\theta}\Big|\langle\tilde{\theta},\textbf{\textit{y}}\rangle\leq b_{1},\langle\tilde{\theta},\textbf{\textit{y}}\rangle\geq b_{2},\|\tilde{\theta}\|_{2}\leq\rho\right\},
\end{eqnarray*}
where $\rho=\frac{\sqrt{n}}{2}$ and $\left\{\tilde{\theta}\Big|||\tilde{\theta}||_{2}\leq\rho\right\}$ is a circumscribed sphere of $\left\{\tilde{\theta}\Big||\tilde{\theta}_{i}|\leq\frac{1}{2}, i=1,\cdots,n\right\}$.
\end{proof}

\begin{Remark}
In Lemma 3.2,  we estimate the dual solution by employing the circumscribed sphere relaxation of the dual box constraints. Therefore, we  name it as the dual circumscribed sphere technique.
\end{Remark}
From the proof of Lemma 3.2, we get that
\begin{align*}
X_{\cdot j}^{T}\theta^{*}(\lambda)&\leq\underset{\theta\in\Theta_{1}}\max {X_{\cdot j}^{T}\theta}=\underset{\theta\in\Theta_{2}}\max {X_{\cdot j}^{T}\theta}+(\tau-\frac{1}{2})\sum\limits_{i=1}^{n}X_{ij}\\&\leq
P^{+}_{j}(\lambda):=\underset{\theta\in\Theta}\max {X_{\cdot j}^{T}\theta}+(\tau-\frac{1}{2})\sum\limits_{i=1}^{n}X_{ij}
\end{align*}
and
\begin{align*}
-X_{\cdot j}^{T}\theta^{*}(\lambda)&\leq\underset{\theta\in\Theta_{1}}\max {-X_{\cdot j}^{T}\theta}=\underset{\theta\in\Theta_{2}}\max {-X_{\cdot j}^{T}\theta}+(\tau-\frac{1}{2})\sum\limits_{i=1}^{n}X_{ij}
\\ &\leq P^{-}_{j}(\lambda):=\underset{\theta\in\Theta}\max {-X_{\cdot j}^{T}\theta}+(\tau-\frac{1}{2})\sum\limits_{i=1}^{n}X_{ij}.
\end{align*}
Combining these results, we obtain that
$$|X_{\cdot j}^{T}\theta^{*}(\lambda)|\leq\max\{P^{+}_{j}(\lambda),P^{-}_{j}(\lambda)\}.$$

Next, we present the detailed results of $P^{+}_{j}(\lambda)$. Before that, we introduce some new notations.

Denote $z^{*}=\underset{\theta\in\Theta}\max {X_{\cdot j}^{T}\theta}$ and $d$ as the distance between the origin of the coordinates and the hyperplane $\left\{\theta\Big|X_{\cdot j}^{T}\theta=z^{*}\right\}$. Then
$$z^{*}=d\|X_{\cdot j}\|_{2}.$$
%In the following, we replace $\hat{\theta}$ with $\theta$ and calculate the close-form of $P^{+}_{j}(\lambda)$, with the help of the distance between a point and a hyperplane.
Denote $\gamma_{j}$ as the angle between $X_{\cdot j}$ and $\textbf{\textit{y}}$, it is easy to see that $\gamma_{j}\in [0,\pi]$ and
\begin{center}
$\rm{cos}$$\gamma_{j}=\frac{\langle X_{\cdot j},\textbf{y}\rangle}{\|X_{\cdot j}\|_{2}\|\textbf{\textit{y}}\|_{2}}$.
\end{center}
Denote $t_{1}$ as the distance between the origin of the coordinates and $\left\{\theta\Big|\langle \theta,\textbf{\textit{y}}\rangle=b_{1}\right\}$. Then $t_{1}=\frac{|b_{1}|}{\|\textbf{\textit{y}}\|_{2}}$. Denote $t_{2}$ as the distance between the origin of the coordinates and  $\left\{\theta\Big|\langle \theta,\textbf{\textit{y}}\rangle=b_{2}\right\}$. Then $t_{2}=\frac{|b_{2}|}{\|\textbf{\textit{y}}\|_{2}}$.
\begin{Lemma}
For any $\lambda\in(0,\lambda_{max}]$,
$$P^{+}_{j}(\lambda)=d\|X_{\cdot j}\|_{2}+(\tau-\frac{1}{2})\sum\limits_{i=1}^{n}X_{ij}$$ with $d$ defined as follows.
\begin{equation}\label{eq:11}
d=
\begin{cases}
\rho,
&\\ \quad \quad \quad\quad \quad\quad\rm{if}~ sign(b_{2})\frac{t_{2}}{\rho}\leq \rm{cos}\gamma_{j}\leq\frac{t_{1}}{\rho}(b_{1}\geq0)
 \\ \quad \quad\quad\quad\quad\quad\quad ~~\rm{or} -\frac{t_{1}}{\rho}\leq \rm{cos}\gamma_{j}\leq-\frac{t_{2}}{\rho}(b_{1}<0)\\
\frac{t_{1}}{\rm{cos}\gamma_{j}}+\sqrt{\rho^{2}-t_{1}^{2}}\cdot\rm{sin}\gamma_{j}-t_{1}\cdot\rm{tan}\gamma_{j}\cdot\rm{sin}\gamma_{j},
& \\ \quad \quad\quad\quad \quad\quad \rm{if}~ \rm{cos}\gamma_{j}>\frac{t_{1}}{\rho} (b_{1}\geq0)\\
\sqrt{\rho^{2}-t_{2}^{2}}\cdot\rm{sin}\gamma_{j}+t_{2}\cdot\rm{cos}\gamma_{j},
&\\ \quad \quad \quad\quad\quad\quad \rm{if}~  \rm{cos}\gamma_{j}<\frac{t_{2}}{\rho}(b_{1}\geq0, b_{2}\geq0)\\
-\frac{t_{2}}{\rm{cos}\gamma_{j}}+\sqrt{\rho^{2}-t_{2}^{2}}\cdot\rm{sin}\gamma_{j}+t_{2}\cdot\rm{tan}\gamma_{j}\cdot\rm{sin}\gamma_{j},
&\\ \quad \quad \quad\quad\quad\quad \rm{if}~ \rm{cos}\gamma_{j}<-\frac{t_{2}}{\rho}(b_{1}\geq0), \rm{cos}\gamma_{j}<-\frac{t_{1}}{\rho}(b_{1}<0)\\
\sqrt{\rho^{2}-t_{1}^{2}}\cdot\rm{sin}\gamma_{j}-t_{1}\cdot\rm{cos}\gamma_{j},
&\\ \quad \quad \quad\quad\quad\quad \rm{if}~ \rm{cos}\gamma_{j}>\frac{t_{1}}{\rho}(b_{1}<0, b_{2}<0).
\end{cases}
\end{equation}
\end{Lemma}
\begin{proof}
There are three cases of the relationship between the hyperplane $\left\{\theta\Big|X_{\cdot j}^{T}\theta=z^{*}\right\}$ and the sphere $\left\{\theta\Big|||\theta||_{2}\leq\rho\right\}$:

(i) $\left\{\theta\Big|X_{\cdot j}^{T}\theta=z^{*}\right\}$ is tangent to $\left\{\theta\Big|\|\theta\|_{2}\leq\rho\right\}$;

(ii) $\left\{\theta\Big|X_{\cdot j}^{T}\theta=z^{*}\right\}$ meets the intersection between $\left\{\theta\Big|\|\theta\|_{2}\leq\rho\right\}$ and $\left\{\theta\Big|X_{\cdot j}^{T}\theta=b_{1}\right\}$;

(iii) $\left\{\theta\Big|X_{\cdot j}^{T}\theta=z^{*}\right\}$ meets the intersection  between $\left\{\theta\Big|\|\theta\|_{2}\leq\rho\right\}$ and $\left\{\theta\Big|X_{\cdot j}^{T}\theta=b_{2}\right\}$.

From Lemma 3.2, we know that $b_{1}\geq b_{2}$. According to the different values of $b_{1}$ and $b_{2}$, the value of $d$ is presented as follows.

Cases 1: $b_{1}\geq0$ and $b_{2}\geq0$. See Fig. 1 for easy understanding.
\begin{figure}[htbp]
\centering {\includegraphics[width=3.5in]{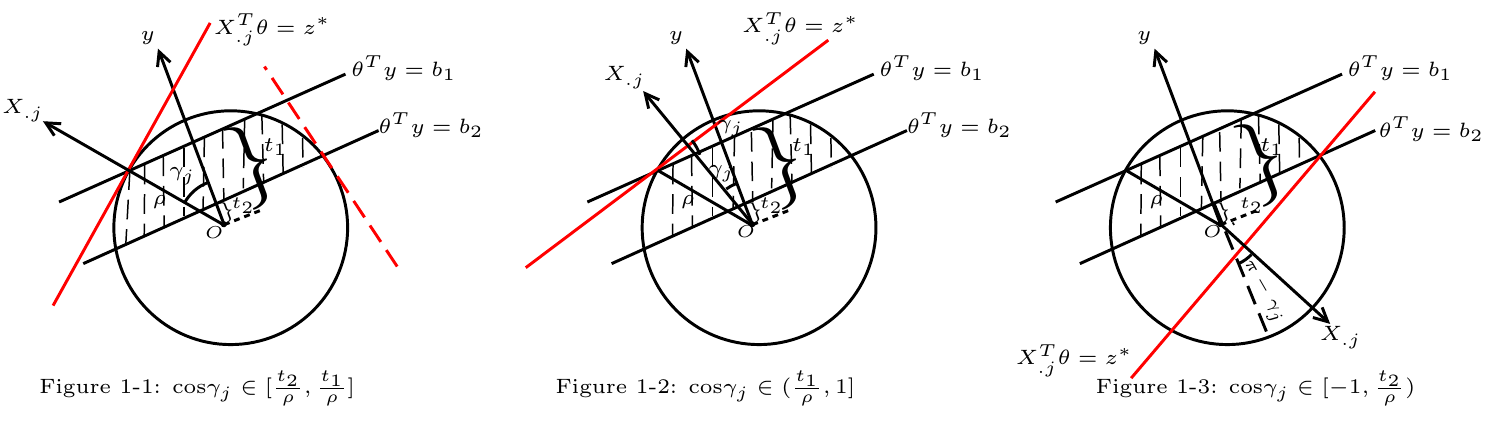}}
\caption{Cases 1: $b_{1}\geq0$ and $b_{2}\geq0$}
\end{figure}

Subcases 1-1: $X_{\cdot j}^{T}\theta=z^{*}$ is tangent to the sphere $\left\{\theta\Big|||\theta||_{2}\leq\rho\right\}$. In this case, the distance between the origin of the coordinates and the hyperplane $\left\{\theta\Big|X_{\cdot j}^{T}\theta=z^{*}\right\}$ is $d=\rho$. The Figure 1-1 in Fig. 1 shows that we use this case if $\frac{t_{2}}{\rho}\leq \rm{cos}$$\gamma_{j}\leq\frac{t_{1}}{\rho}$.

 Subcases 1-2: $X_{\cdot j}^{T}\theta=z^{*}$ meets the intersection between $\left\{\theta\Big|||\theta||_{2}\leq\rho\right\}$ and  $X_{\cdot j}^{T}\theta=b_{1}$. Here,
\begin{center}
$d=\frac{t_{1}}{\rm{cos}\gamma_{j}}+\sqrt{\rho^{2}-t_{1}^{2}}\cdot\rm{sin}\gamma_{j}-t_{1}\cdot\rm{tan}\gamma_{j}\cdot\rm{sin}\gamma_{j}$.
\end{center}
The result of Figure 1-2 in Fig. 1 shows that $\rm{cos}$$\gamma_{j}>\frac{t_{1}}{\rho}$.

Subcases 1-3: $X_{\cdot j}^{T}\theta=z^{*}$ meets the intersection between $\left\{\theta\Big|||\theta||_{2}\leq\rho\right\}$ and  $X_{\cdot j}^{T}\theta=b_{2}$. Similarly,
\begin{center}
$d=\sqrt{\rho^{2}-t_{2}^{2}}\cdot\rm{sin}\gamma_{j}+t_{2}\cdot\rm{cos}\gamma_{j}$.
\end{center}
From Figure 1-3 in Fig. 1, we know that $\rm{cos}$$\gamma_{j}<\frac{t_{2}}{\rho}$ in this case.

With the similar analysis, we can obtain the other cases as follows.\\
Cases 2: $b_{1}\geq0$ and $b_{2}<0$. See Fig. 2 for easy understanding.
\begin{figure}[htbp]
\centering {\includegraphics[width=3.5in]{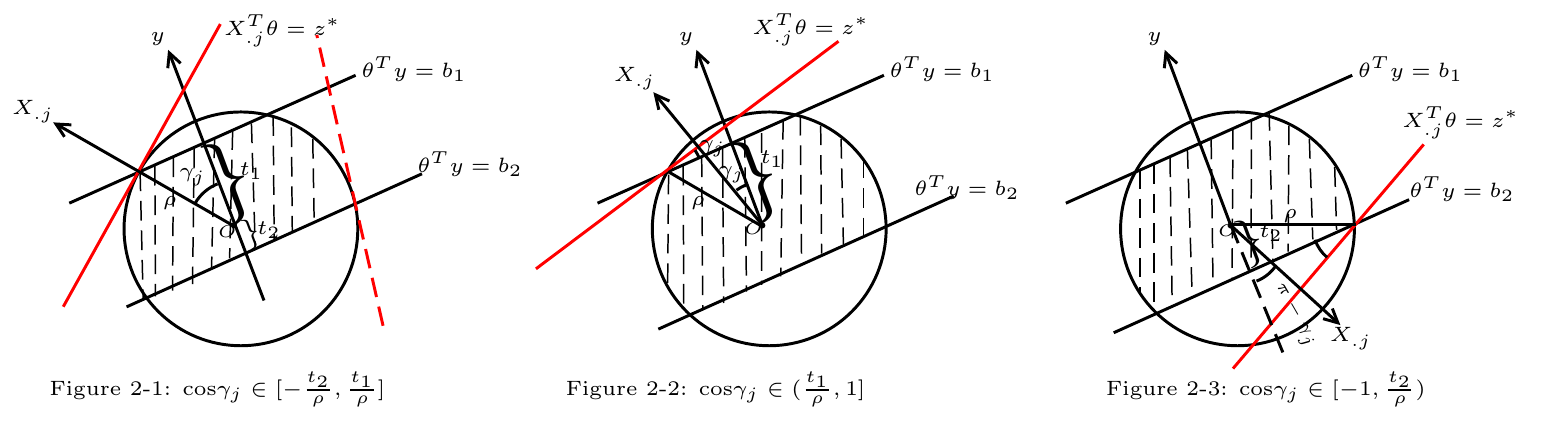}}
\caption{Cases 2: $b_{1}\geq0$ and $b_{2}<0$}
\end{figure}

Subcases 2-1: $d=\rho$ when $-\frac{t_{2}}{\rho}\leq \rm{cos}$$\gamma_{j}\leq\frac{t_{1}}{\rho}$.

Subcases 2-2: If $\rm{cos}$$\gamma_{j}>\frac{t_{1}}{\rho}$,
\begin{center}
$d=\frac{t_{1}}{\rm{cos}\gamma_{j}}+\sqrt{\rho^{2}-t_{1}^{2}}\cdot\rm{sin}\gamma_{j}-t_{1}\cdot\rm{tan}\gamma_{j}\cdot\rm{sin}\gamma_{j}$.
\end{center}

Subcases 2-3: If $\rm{cos}$$\gamma_{j}<-\frac{t_{2}}{\rho}$
\begin{center}
$d=-\frac{t_{2}}{\rm{cos}\gamma_{j}}+\sqrt{\rho^{2}-t_{2}^{2}}\cdot\rm{sin}\gamma_{j}+t_{2}\cdot\rm{tan}\gamma_{j}\cdot\rm{sin}\gamma_{j}$.
\end{center}
Cases 3: $b_{1}<0$ and $b_{2}<0$. See Fig. 3 for easy understanding.

\begin{figure}[htbp]
\centering {\includegraphics[width=3.5in]{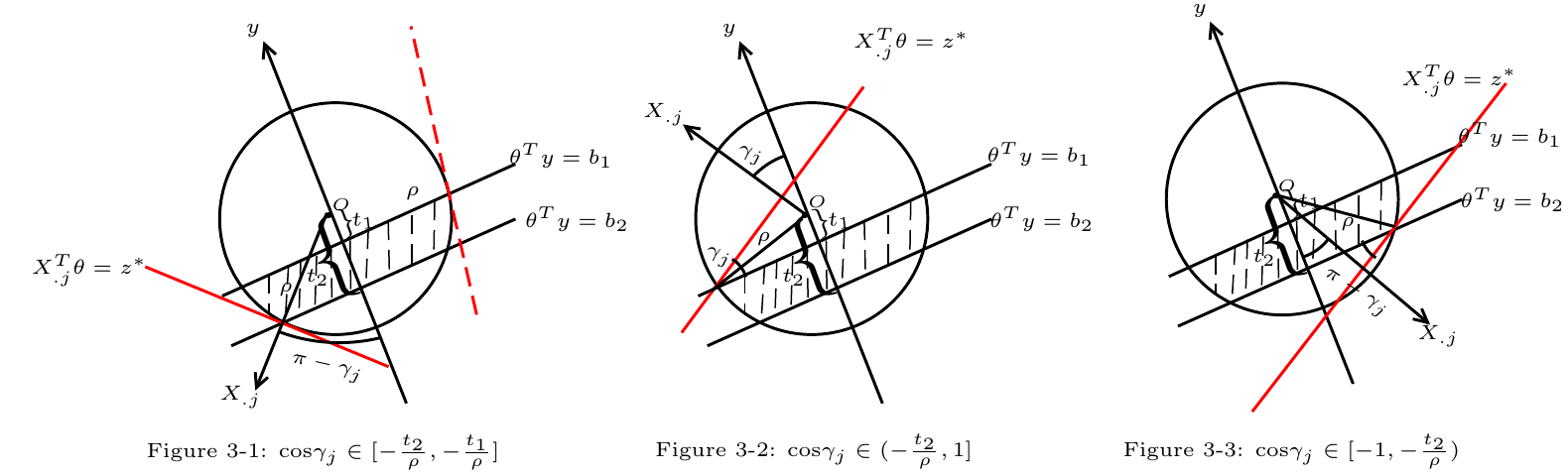}}
\caption{Cases 3: $b_{1}<0$ and $b_{2}<0$}
\end{figure}
Subcases 3-1: $d=\rho$ when $-\frac{t_{1}}{\rho}\leq \rm{cos}$$\gamma_{j}\leq-\frac{t_{2}}{\rho}$.

Subcases 3-2: If $\rm{cos}$$\gamma_{j}>-\frac{t_{2}}{\rho}$,
\begin{center}
$d=\sqrt{\rho^{2}-t_{1}^{2}}\cdot\rm{sin}\gamma_{j}-t_{1}\cdot\rm{cos}\gamma_{j}$.
\end{center}

 Subcases 3-3: If $\rm{cos}$$\gamma_{j}<-\frac{t_{1}}{\rho}$,
\begin{center}
$d=-\frac{t_{2}}{\rm{cos}\gamma_{j}}+\sqrt{\rho^{2}-t_{2}^{2}}\cdot\rm{sin}\gamma_{j}+t_{2}\cdot\rm{tan}\gamma_{j}\cdot\rm{sin}\gamma_{j}$.
\end{center}

To sum up, we can obtain the equation (\ref{eq:11}). Therefore,
$$P^{+}_{j}(\lambda)=d\|X_{\cdot j}\|_{2}+(\tau-\frac{1}{2})\sum\limits_{k=1}^{n}X_{kj}$$ with $d$ defined in  (\ref{eq:11}).
\end{proof}
The detailed result of $P^{-}_{j}(\lambda)$ can be gotten in the similar way by defining $\gamma_{j}$ as the angle between $-X_{\cdot j}$ and $\textbf{\textit{y}}$.
\begin{Lemma}
For any $\lambda\in(0,\lambda_{max}]$,
$$P^{-}_{j}(\lambda)=\tilde{d}\|X_{\cdot j}\|_{2}+(\tau-\frac{1}{2})\sum\limits_{i=1}^{n}X_{ij}$$ with $\tilde{d}$ having the form of (\ref{eq:11}), where $\gamma_{j}$ is the angle between $-X_{\cdot j}$ and $\textbf{\textit{y}}$.
\end{Lemma}

In Lemma 3.3 and Lemma 3.4,  $P^{+}_{j}(\lambda)$ and $P^{-}_{j}(\lambda)$ show the closed-form expressions of given data, which can be easily computed. Combing these lemma and Theorem 3.2, we obtain the implementable safe feature screening rule for the  $\ell_{1}$-norm  quantile regression in the next theorem.
\begin{Theorem}
For given $\lambda\in(0,\lambda_{max}]$ and $j\in \{1,\cdots,p\}$. If $$\max\left\{P^{+}_{j}(\lambda),P^{-}_{j}(\lambda)\right\}<\lambda,$$
 then $\beta_{j}^{*}(\lambda)=0$, which means that  $X_{\cdot j}$ is uncorrelated to the $\tau_{th}$ quantile of $\textbf{y}$ and $X_{\cdot j}$ can be eliminated.
\end{Theorem}
\begin{Remark}
Theorem 3.1 can be extended to the weighted $\ell_{1}$-norm  quantile regression, which is given as
\begin{equation}\label{eq:14}
\underset{\beta}\min\sum\limits_{i=1}^{n}\rho_{\tau}(y_{i}-\textbf{x}^{T}_{i}\beta)+\lambda\sum\limits_{j=1}^{p}\|\omega\odot\beta\|,
\end{equation}
where $\omega=(\omega_{1},\omega_{2},\cdots,\omega_{p})^{T}\in\mathbb{R}^{p}$ and $\omega_{j}>0$ holds for all $j$. When the weight $\omega=(1,1,\cdots,1)^{T}$ in (\ref{eq:14}), it degrades to the $\ell_{1}$-norm quantile regression. For the model (\ref{eq:14}), $$\tilde{\lambda}_{max}=\frac{\|\Delta\|_{\infty}}{\underset{j}\min~\omega_{j}},$$
 where $\Delta$ is defined in (\ref{eq:new2}). With Lemma 3.3 and Lemma 3.4, we can obtain the next screening rule for the weighted $\ell_{1}$-norm  quantile regression.

For given $\lambda\in(0,\tilde{\lambda}_{max}]$ and $j\in \{1,\cdots,p\}$. If $$\max\left\{P^{+}_{j}(\lambda),P^{-}_{j}(\lambda)\right\}<\lambda\omega_{j},$$
 then $\beta_{j}^{*}(\lambda)=0$, which means that  $X_{\cdot j}$ is uncorrelated to the $\tau_{th}$ quantile of $\textbf{\textit{y}}$ and $X_{\cdot j}$ can be eliminated.
\end{Remark}

\section{Numerical results}
In order to evaluate the performance of the $\ell_{1}$-norm  quantile regression screening rule in Section 3, we do some numerical experiments in this section. All experiments are performed on the MATLAB R2018b with Intel(R) Core(TM) i5-8250U 1.60 CPU and 8G RAM.

For each data set, we run the proximal ADMM (Gu et al. \cite{G18}) along a sequence of 100 tuning parameters equally spaced on the $\lambda/\lambda_{max}$ from 0.01 to 1.  Same as the most papers (Tibshirani et al. \cite{T12}, Wang et al. \cite{W15}, Wang et al. \cite{W15b},  Ndiaye et al. \cite{N17}, Xiang et al. \cite{X17} and so on), we use two quantities to measure the screening rule, that are rejection ratio and speedup. The rejection ratio is defined as
$R=\frac{N_{f}}{N_{s}}.$ Under different tuning parameter $\lambda$, $N_{s}$ denotes the number of discarded features by the screening rule and $N_{f}$ denotes the  actual number of features with zero coefficient. This ratio measures the efficiency of the screening rule.  The speedup is defined as $S=\frac{T_{f}}{T_{s}},$
where $T_{s}$ and $T_{f}$ indicate the calculation time of solver with and without screening rule respectively.  This quantity, as the name illustrated, indicates  the reduced computational time because of the screening rule. The larger the rejection ratio and speedup are, the more efficient the screening rule is.
\subsection{Simulation}
We evaluate the screening rule on some simulation data, which are generated from the true linear regression
$$\textbf{\textit{y}}=X\beta^{*}+\epsilon$$
with $n=100$ and $p\in\left\{5000, 10000, 15000\right\}$. In the true model, $X$ is generated from the multivariate normal distribution with mean vector $\mu$ and covariance matrix $\Sigma$. Here, we consider two cases, that are $\Sigma^{1}=I_{p}$ and $\Sigma^{2}_{ij}=0.5^{|i-j|}$ as in Belloni\cite{B11}, Li and Zhu\cite{L08} and so on. To show the difference between features (every column of $X$ presents a feature), we set $\mu$ as $\mu(3:7)=10$,
$\mu(70:90)=5$ and $\mu([\frac{p}{2}]:[\frac{2p}{3}])=-2$. The true coefficient vector $\beta^{*}$ is set as follows.
$$\beta^{*}=(2,0,1.5,0,0,0.8,0,0,1,0,1.75,0,0,0.75,0,0,0.3,0_{p-16}).$$
The error $\epsilon\sim t(4)$. See, e.g.,  Fan \cite{F14} and Gu \cite{G18}.
\begin{figure}[htbp]
\begin{minipage}[t]{1\linewidth}
\centering
\includegraphics[width=2.3in]{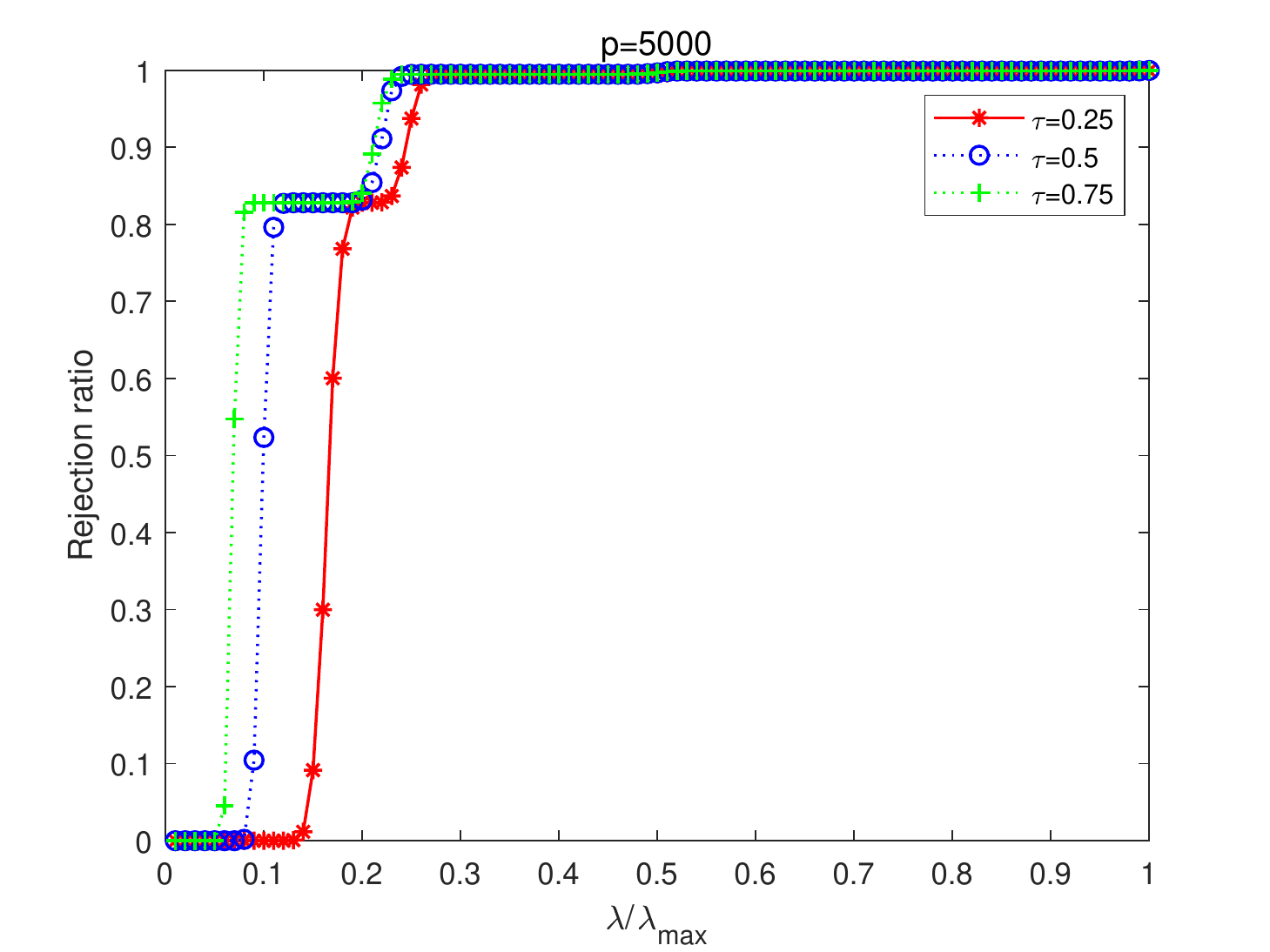}
\end{minipage}
\begin{minipage}[t]{1\linewidth}
\centering
\includegraphics[width=2.3in]{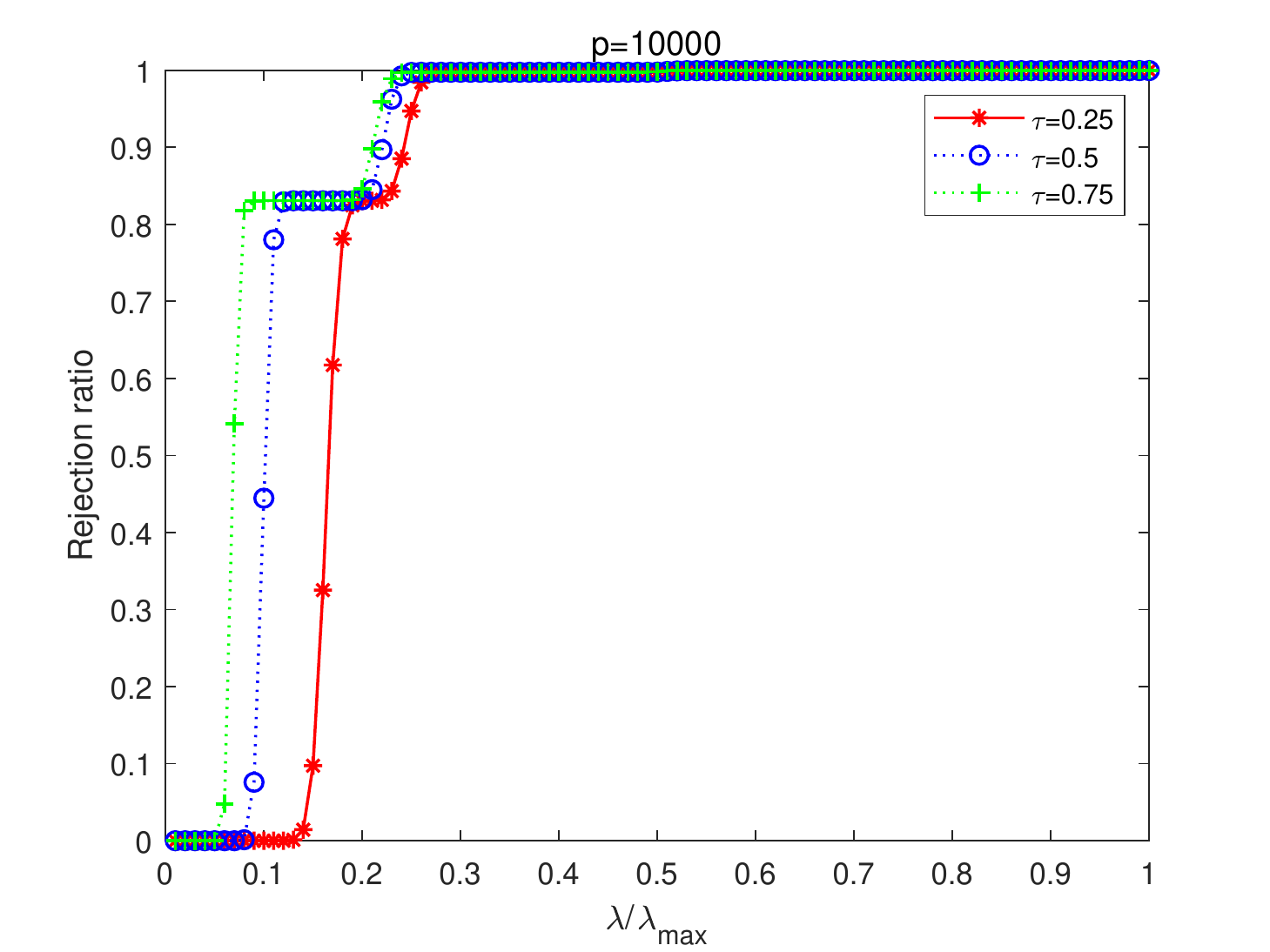}
\end{minipage}
\begin{minipage}[t]{1\linewidth}
\centering
\includegraphics[width=2.3in]{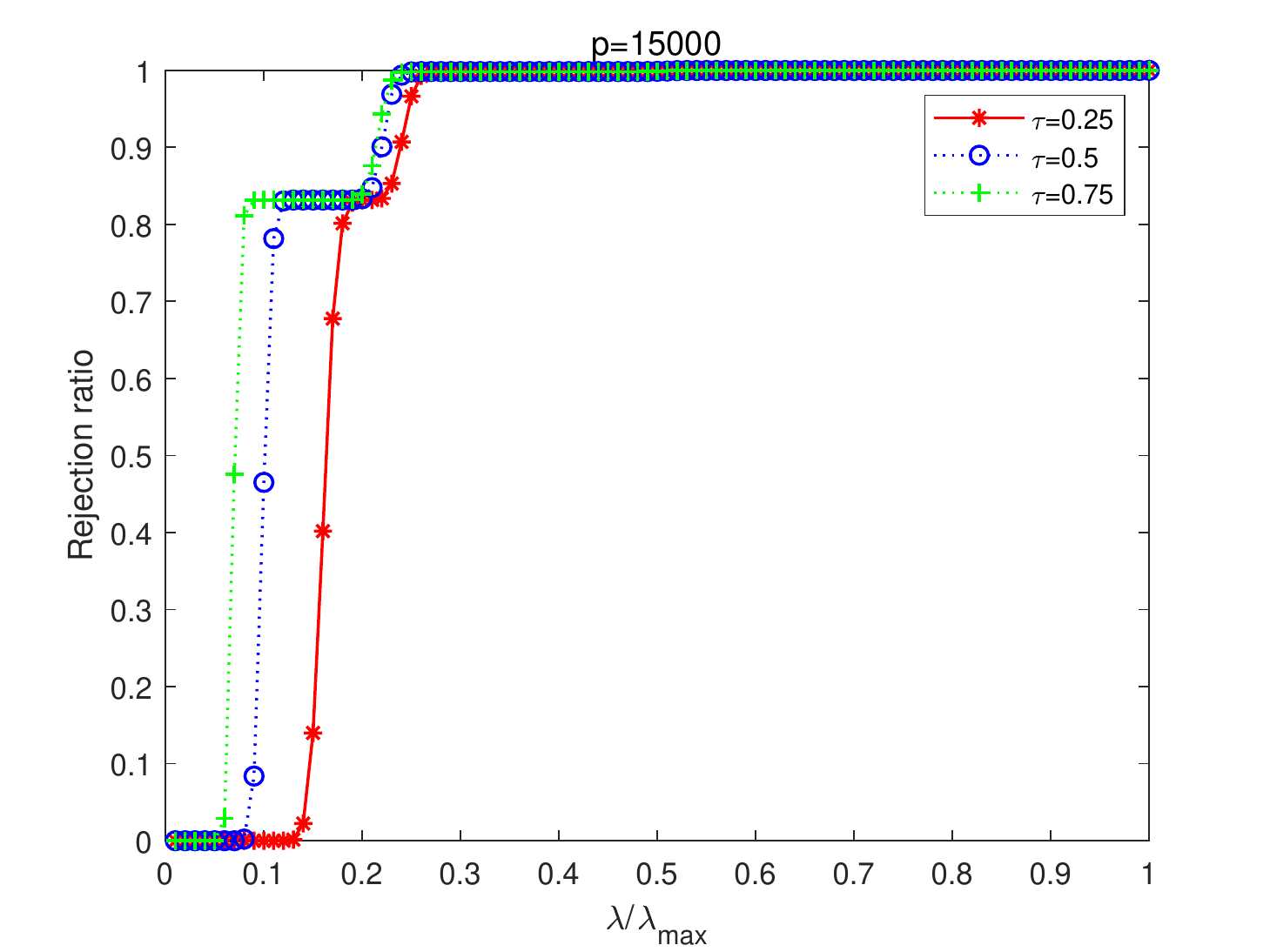}
\end{minipage}
\caption{The rejection ratio on $\Sigma^{1}$ with different $p$ and $\tau$. There are slight differences between $\Sigma^{1}$ and $\Sigma^{2}$ on the rejection ratio, so we omit these results on $\Sigma^{2}$.}
\end{figure}

From Fig. 4, for the simulated data, we conclude that $\tau=0.75$ performs best in rejection ratio. In this case, our screening rule identifies over 80 percent inactive features.  In addition, the gap of rejection ratios of $\tau=0.75$ and $\tau=0.5$  is very small.  In order to clearly present the different performance of our screening rule under different $p$, $\tau$ and $\Sigma$, we report the speedup in the TABLE 1.
\begin{table}[htbp]
\caption{ In this table, we set the sample size $n=100$ and report the speedup of different datasets under different $p$, $\Sigma$ and $\tau$. These reported speedup are the means of 10 simulation results. Here, we also report the standard deviations (sd) of speedup values.}
\centering
\begin{tabular}{|c|c|c|c|c|c|c|}
\hline
\multicolumn{2}{|c|}{$p,\Sigma$}                        & $\tau$ & $T_{f}(s)$ & $T_{s}(s)$ & speedup & sd \\ \hline
\multicolumn{2}{|c|}{\multirow{3}{*}{5000,$\Sigma^{1}$}}  & 0.25                & 30.0053  & 9.3310   & 3.2157  & 0.0299            \\
\multicolumn{2}{|c|}{}                           & 0.50                & 29.8827  & 6.7459   & 4.4298  & 0.0589            \\
\multicolumn{2}{|c|}{}                           & 0.75                & 22.2410  & 5.0821   & 4.3763  & 0.0652            \\ \hline
\multicolumn{2}{|c|}{\multirow{3}{*}{10000,$\Sigma^{1}$}} & 0.25       &  112.2569  &29.9057 &  3.7537 &   0.0340          \\
\multicolumn{2}{|c|}{}                           & 0.50                & 112.5967 & 20.5164   & 5.4881 & 0.0526            \\
\multicolumn{2}{|c|}{}                           & 0.75                & 84.6651  & 14.0857   & 6.0107 & 0.0574            \\ \hline
\multicolumn{2}{|c|}{\multirow{3}{*}{15000,$\Sigma^{1}$}} & 0.25       & 318.0094 & 78.7775  & 4.0368 &  0.1530           \\
\multicolumn{2}{|c|}{}                           & 0.50                & 250.6742  & 41.8572   & 5.9888 &  0.0495                 \\
\multicolumn{2}{|c|}{}                           & 0.75                &  186.1347 & 27.2497  & 6.8307  &  0.0428                 \\ \hline
\multicolumn{2}{|c|}{\multirow{3}{*}{5000,$\Sigma^{2}$}}  & 0.25                & 29.2768  & 9.2876   & 3.1522  &0.0420             \\
\multicolumn{2}{|c|}{}                           & 0.50                         &29.3762   & 6.7878  &4.3278   &  0.0749           \\
\multicolumn{2}{|c|}{}                           & 0.75                         &21.8868  & 5.1286  & 4.2676  & 0.0723           \\ \hline
\multicolumn{2}{|c|}{\multirow{3}{*}{10000,$\Sigma^{2}$}} & 0.25                & 111.4912  & 30.4614 & 3.6601  & 0.0479                  \\
\multicolumn{2}{|c|}{}                           & 0.50                & 110.9368& 20.7868   &5.3369  & 0.0341            \\
\multicolumn{2}{|c|}{}                           & 0.75                &  82.8418 &14.4750   &5.7231  &  0.0374           \\ \hline
\multicolumn{2}{|c|}{\multirow{3}{*}{15000,$\Sigma^{2}$}} & 0.25       &302.2332          & 77.3164         & 3.9090        & 0.0472                 \\
\multicolumn{2}{|c|}{}                           & 0.50                & 312.9282         &  53.6676        & 5.8309       &  0.1181                \\
\multicolumn{2}{|c|}{}                           & 0.75                & 232.9686        & 35.3582        &  6.5888       & 0.2220                 \\ \hline
\end{tabular}
\end{table}

According to the results in  TABLE 1, we have some conclusions. Firstly, our screening rule shrinks the computational time in different degrees. The smallest speedup is 3.1522 and the largest is 6.8307.  Secondly, with the same sample size, the speedup value increases with $p$ increasing. From TABLE 1, we know that $T_{f}$ rapidly increase when $p$ increases, while $T_{s}$ increases slightly.   Thirdly, the standard variances of the speedup values are relatively small. These standard variances show that our screening rule has a stable performance on speedup. Finally, there are slightly differences between  $\Sigma=\Sigma^{1}$ and  $\Sigma=\Sigma^{2}$, which means our screening rule is stable on data sets with uncorrelated features or correlated features.

%\begin{figure}[htbp]
%\centering
%\includegraphics[width=3.5in]{speedup-simulation}
%\caption{The speedup of different datasets under different p and $\tau$.}
%\end{figure}
\subsection{Real data}

We analyze some real data sets, named Colon\cite{A99}, Srbct\cite{K01}, Lymphoma\cite{A00}, Brain\cite{P02}, Prostate\cite{B10}, Leukemia , Reuters21578, 20Newsgroups and TDT2. Obviously, not all feature have connections with the $\tau_{th}$ quantile of the response variable. Therefore, we do not need to input all these features when calculating their relationships, which means the screening rule is needed before dealing with this issue. We apply the screening rule to $\tau=0.25$, $0.5$ and $0.75$ as Scheetz et al. \cite{S06}, Peng and Wang \cite{P15}, Gu et al. \cite{G18} and so on.
%Next, we introduce these data sets briefly.
%\begin{itemize}
%\item Leukemia: This dataset consists of gene expression data for 7128 genes of 38 samples.
%\item Reuters21578: The preprocessed Reuters-21578 corpus have 8293 documents, and every corpus has 18933 distinct terms.
%\item 20Newsgroups: The training data  of the 20Newsgroups  consists of 11314  documents, and every data has 26214 features.
%\item TDT2: The TDT2 corpus consist of data collected during the first half of 1998 and taken from 6 sources, and every data has 36771 features. The largest 30 categories were kept, leaving 9,394 documents in total.
%\end{itemize}
\begin{figure}[htbp]

\subfigure[Colon $(62\times2000)$]{
\begin{minipage}[t]{0.8\linewidth}
\centering
\includegraphics[width=2.2in]{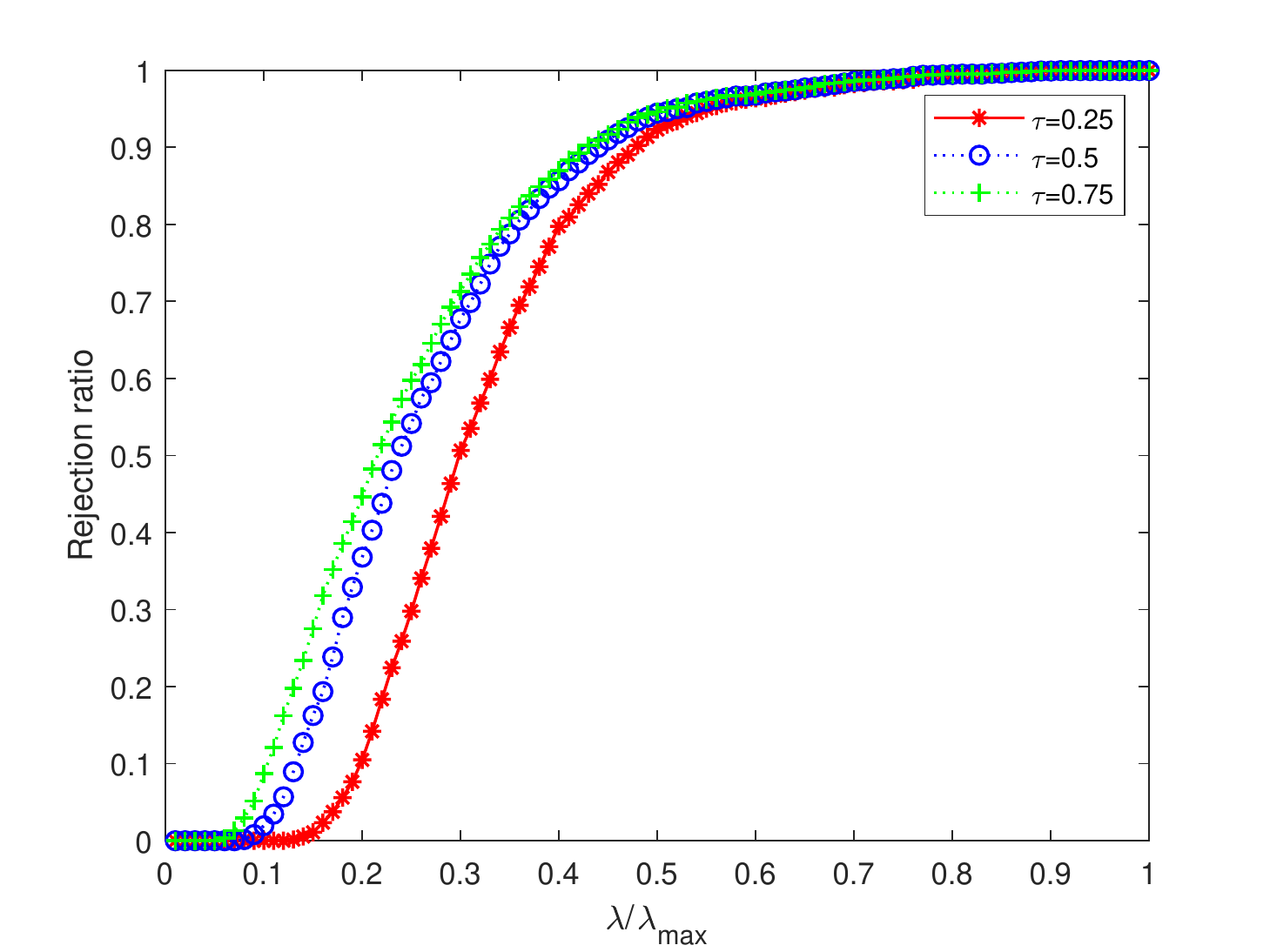}
%\caption{fig2}
\end{minipage}
}%

\subfigure[Srbct $(63\times2308)$]{
\begin{minipage}[t]{0.8\linewidth}
\centering
\includegraphics[width=2.2in]{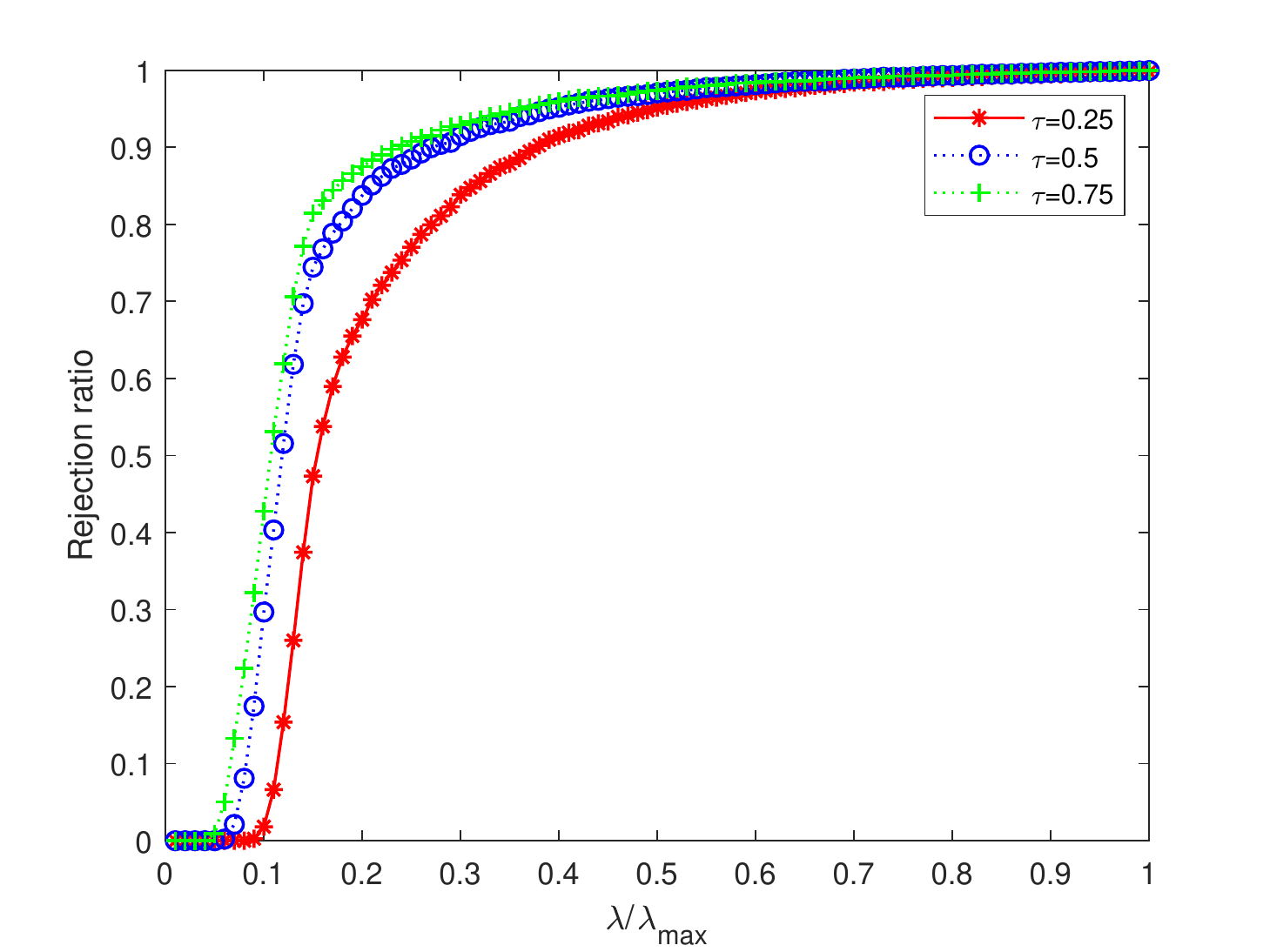}
%\caption{fig2}
\end{minipage}
}

\subfigure[Lymphoma $(62\times4026)$]{
\begin{minipage}[t]{0.8\linewidth}
\centering
\includegraphics[width=2.2in]{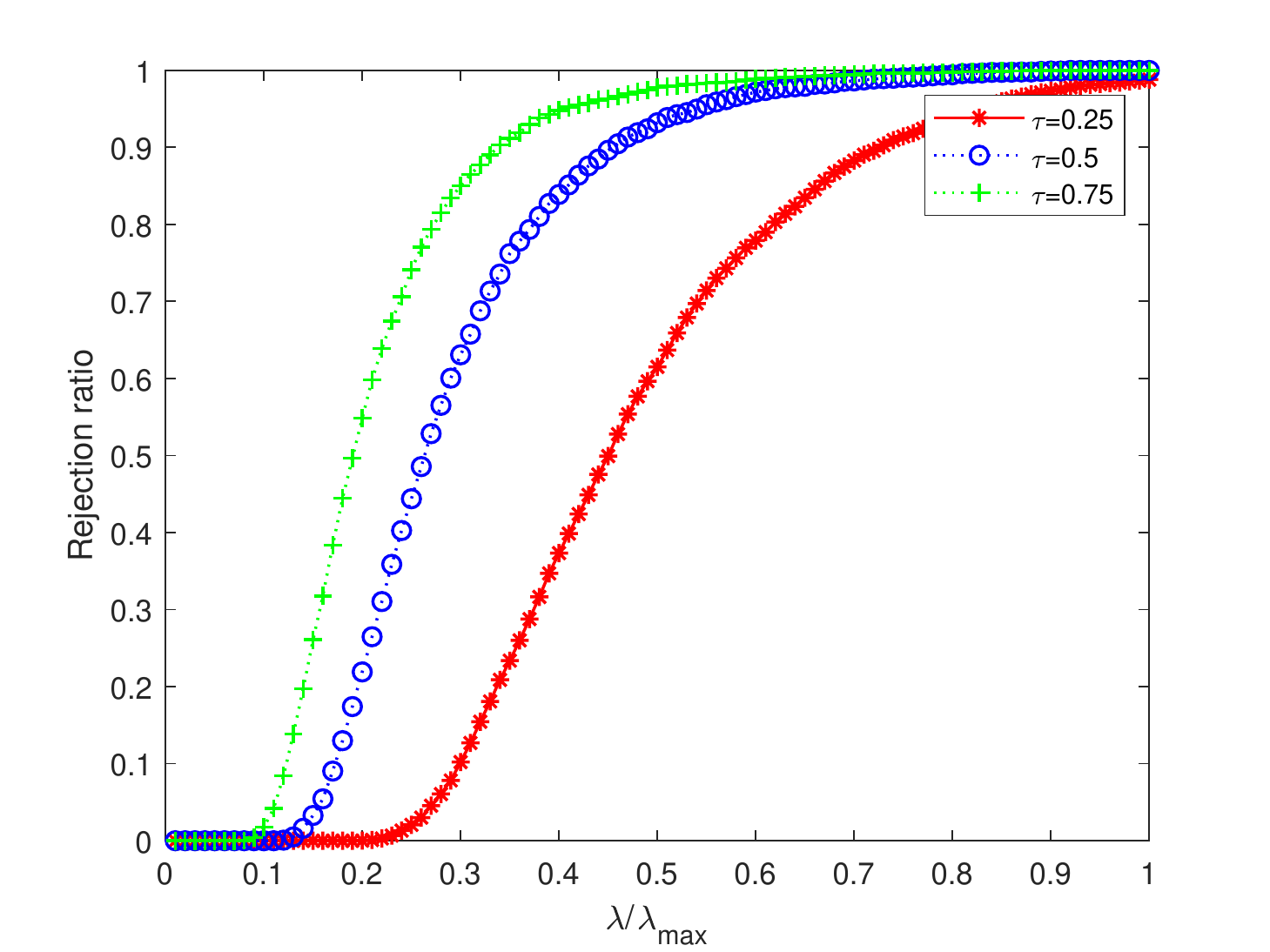}
%\caption{fig2}
\end{minipage}
}

\subfigure[Brain $(42\times5597)$]{
\begin{minipage}[t]{0.8\linewidth}
\centering
\includegraphics[width=2.2in]{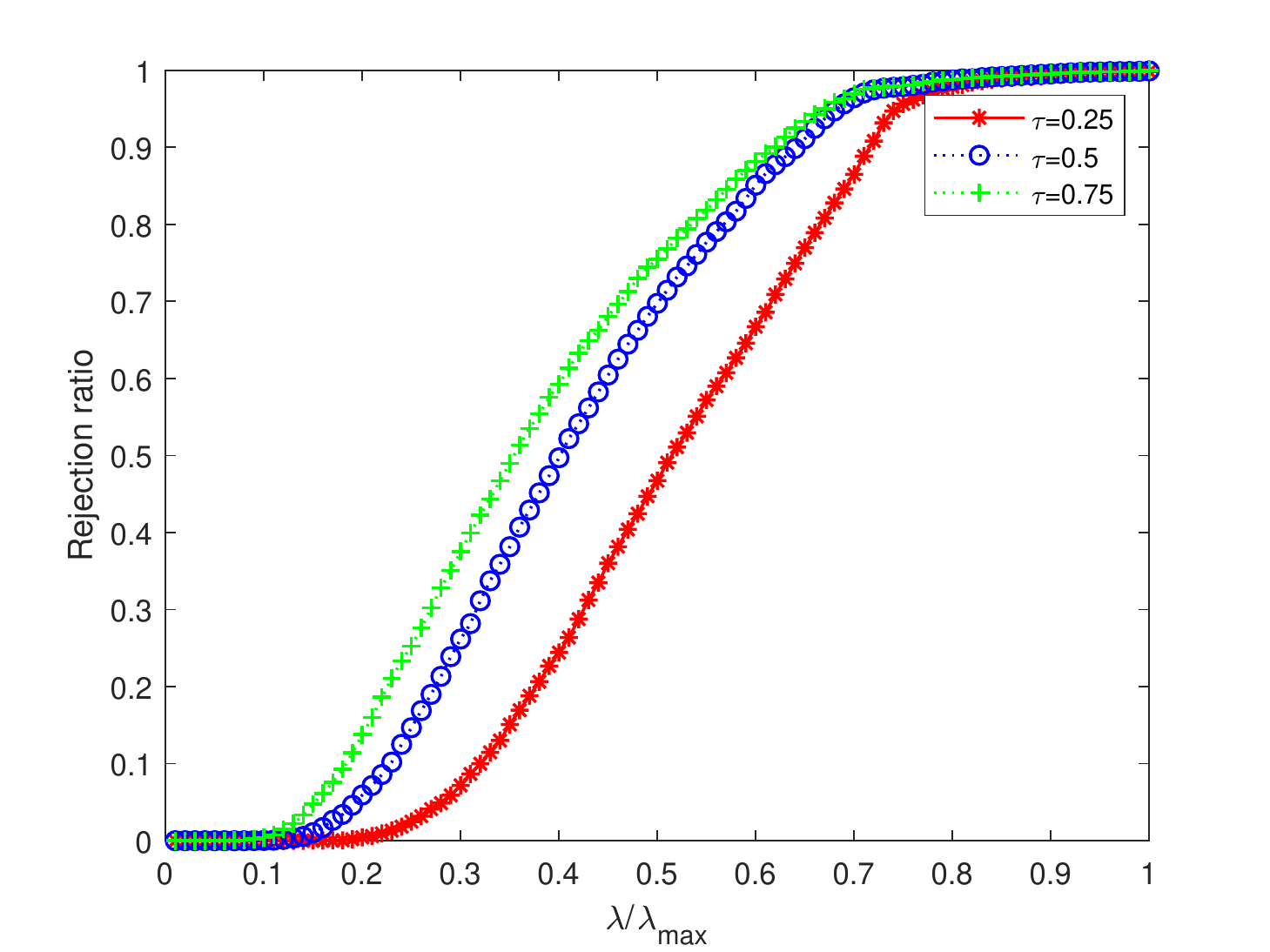}
%\caption{fig2}
\end{minipage}
}

\subfigure[Prostate $(102\times6033)$]{
\begin{minipage}[t]{0.8\linewidth}
\centering
\includegraphics[width=2.2in]{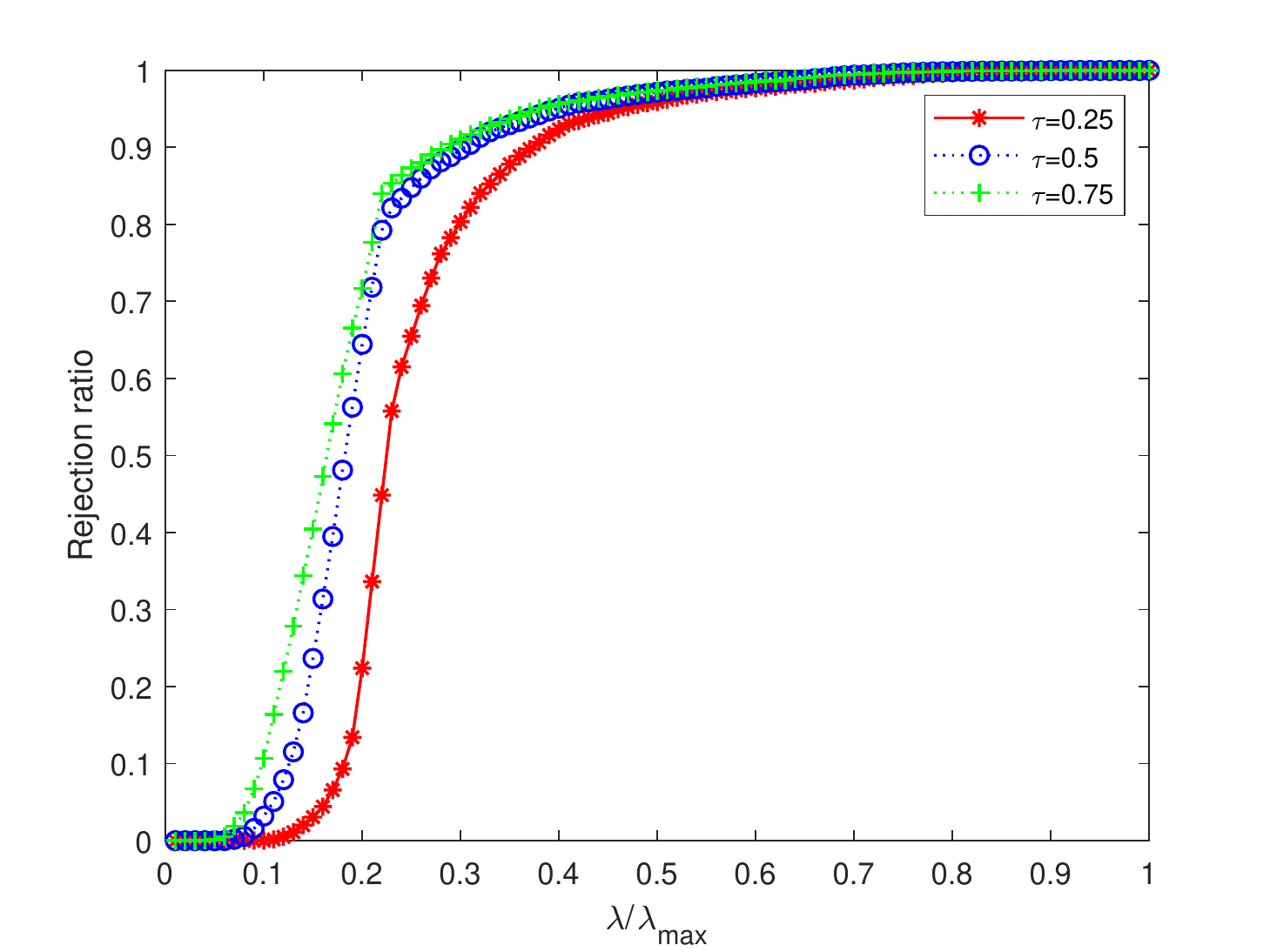}
%\caption{fig2}
\end{minipage}
}
\caption{The rejection ratio under different data sets.}
\end{figure}
\begin{figure}[htbp]
\subfigure[Leukemia $(38\times7128)$]{
\begin{minipage}[t]{0.8\linewidth}
\centering
\includegraphics[width=2.2in]{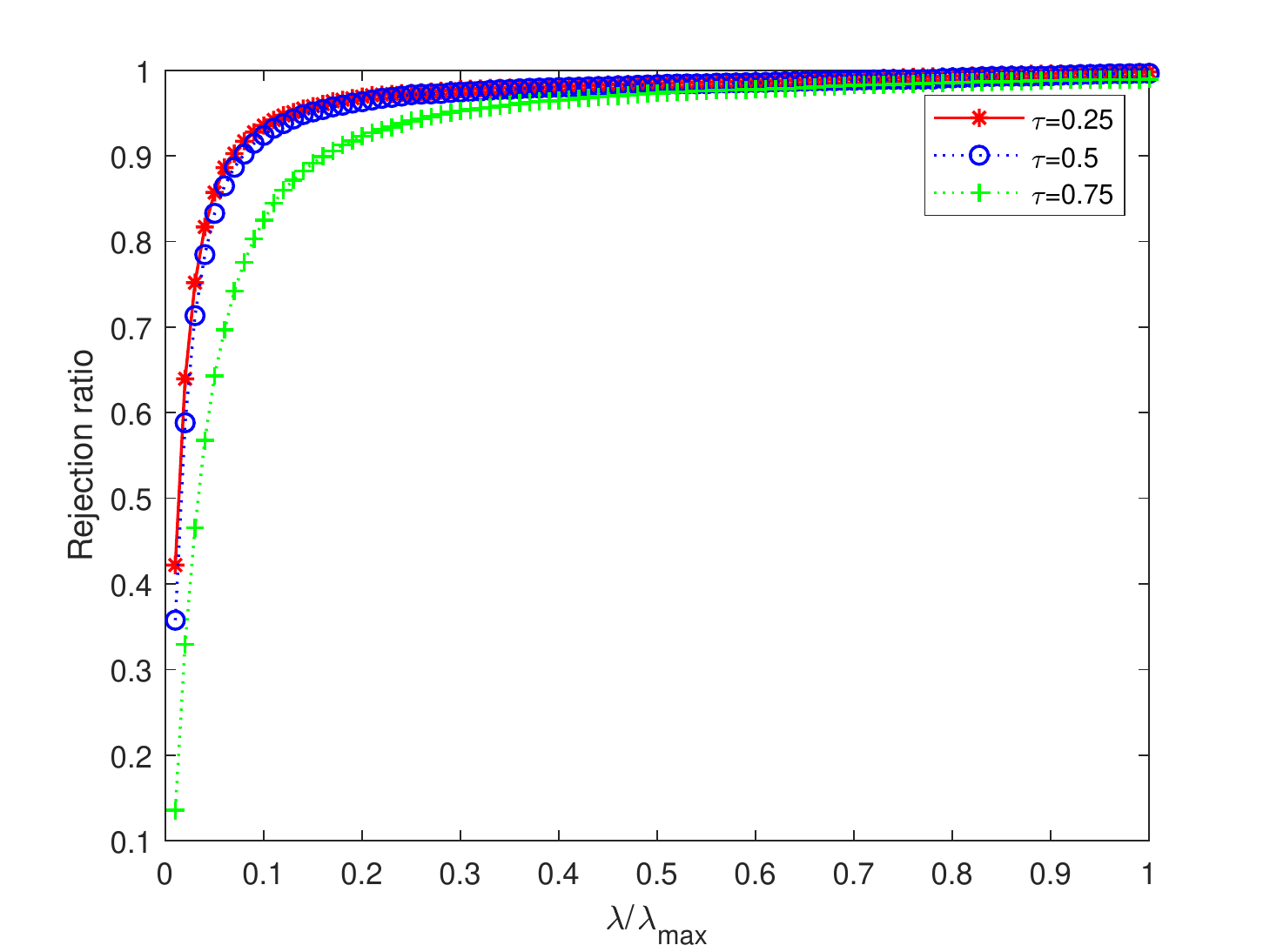}
%\caption{fig2}
\end{minipage}
}

\subfigure[Reuters21578 $(8293\times18933)$]{
\begin{minipage}[t]{0.8\linewidth}
\centering
\includegraphics[width=2.2in]{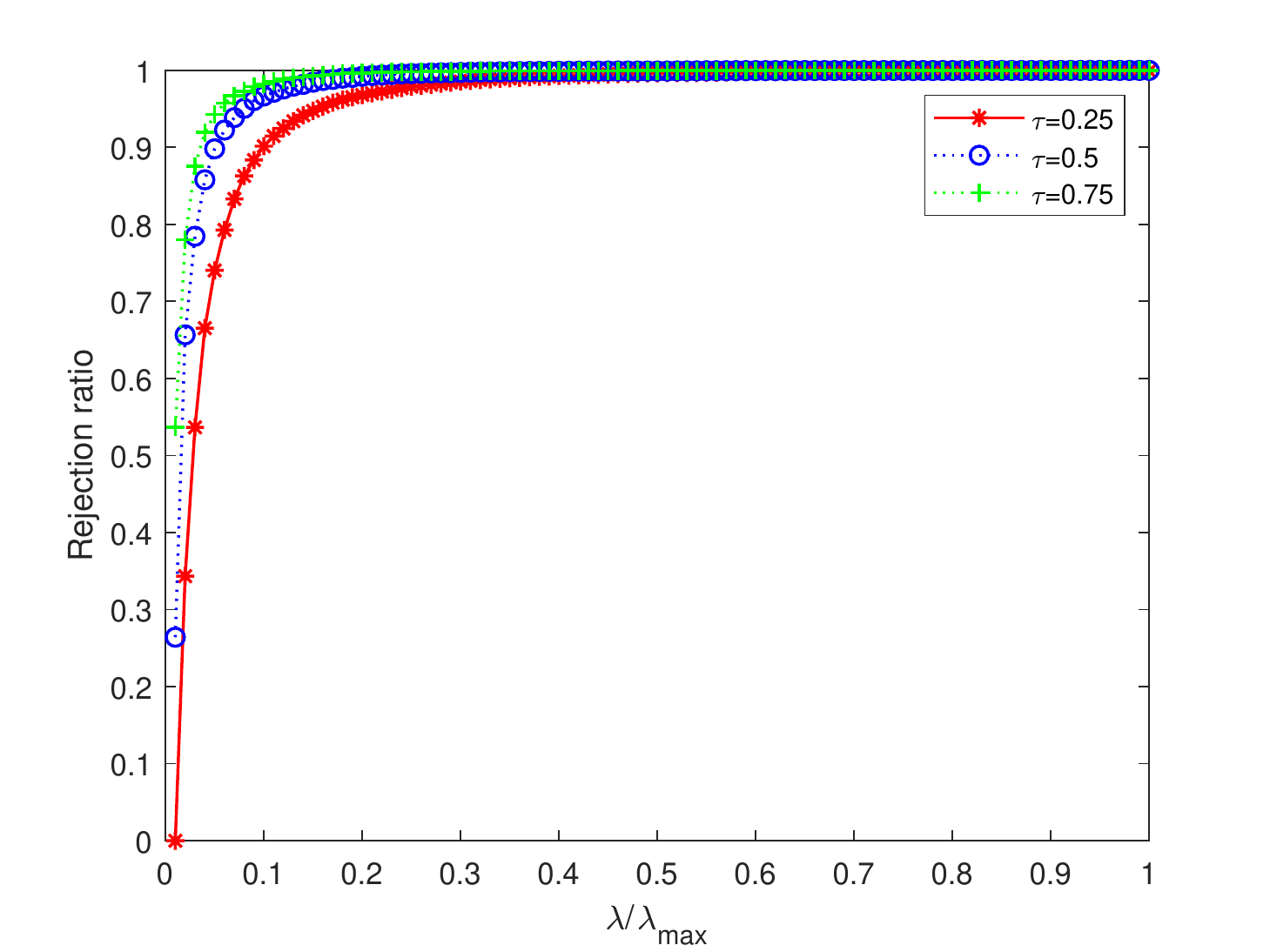}
%\caption{fig2}
\end{minipage}
}%

\subfigure[20Newsgroups $(11314\times26214)$]{
\begin{minipage}[t]{0.8\linewidth}
\centering
\includegraphics[width=2.2in]{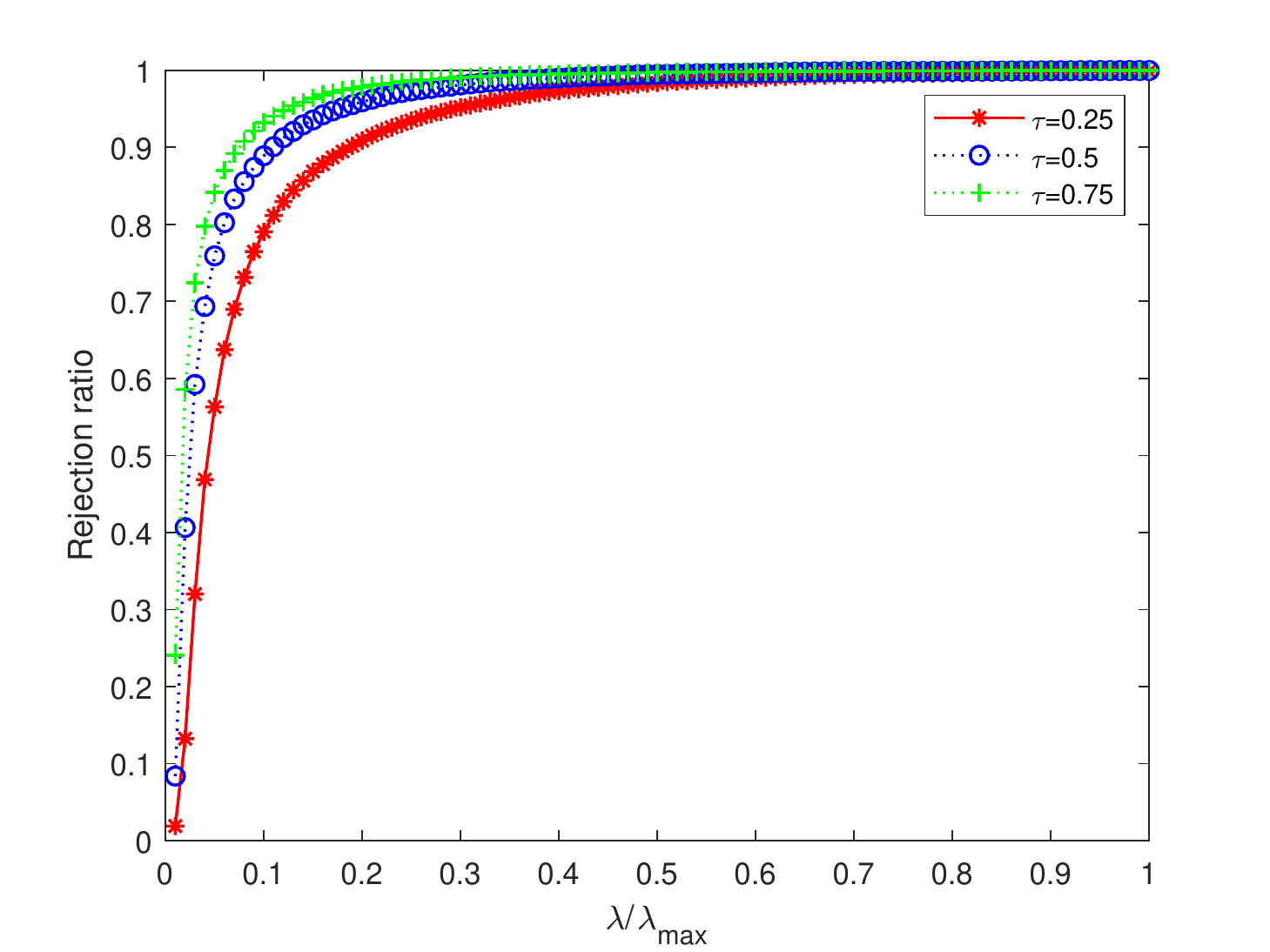}
%\caption{fig2}
\end{minipage}
}

\subfigure[TDT2 $(9394\times36771)$]{
\begin{minipage}[t]{0.8\linewidth}
\centering
\includegraphics[width=2.2in]{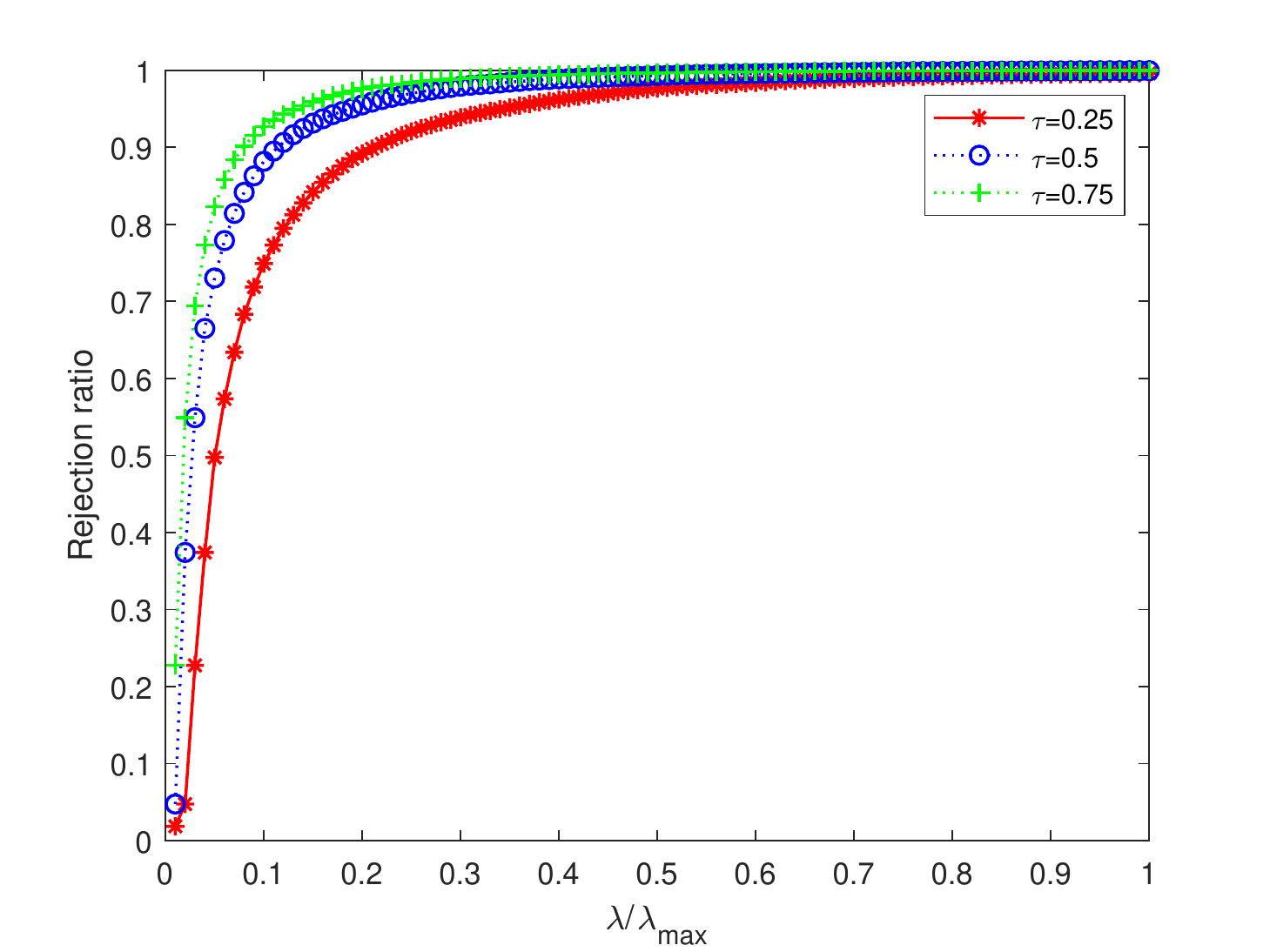}
%\caption{fig2}
\end{minipage}
}
\caption{The rejection ratio under different data sets.}
\end{figure}

\begin{figure}[htbp]

\subfigure{
\begin{minipage}[t]{1\linewidth}
\centering
\includegraphics[width=3.5in]{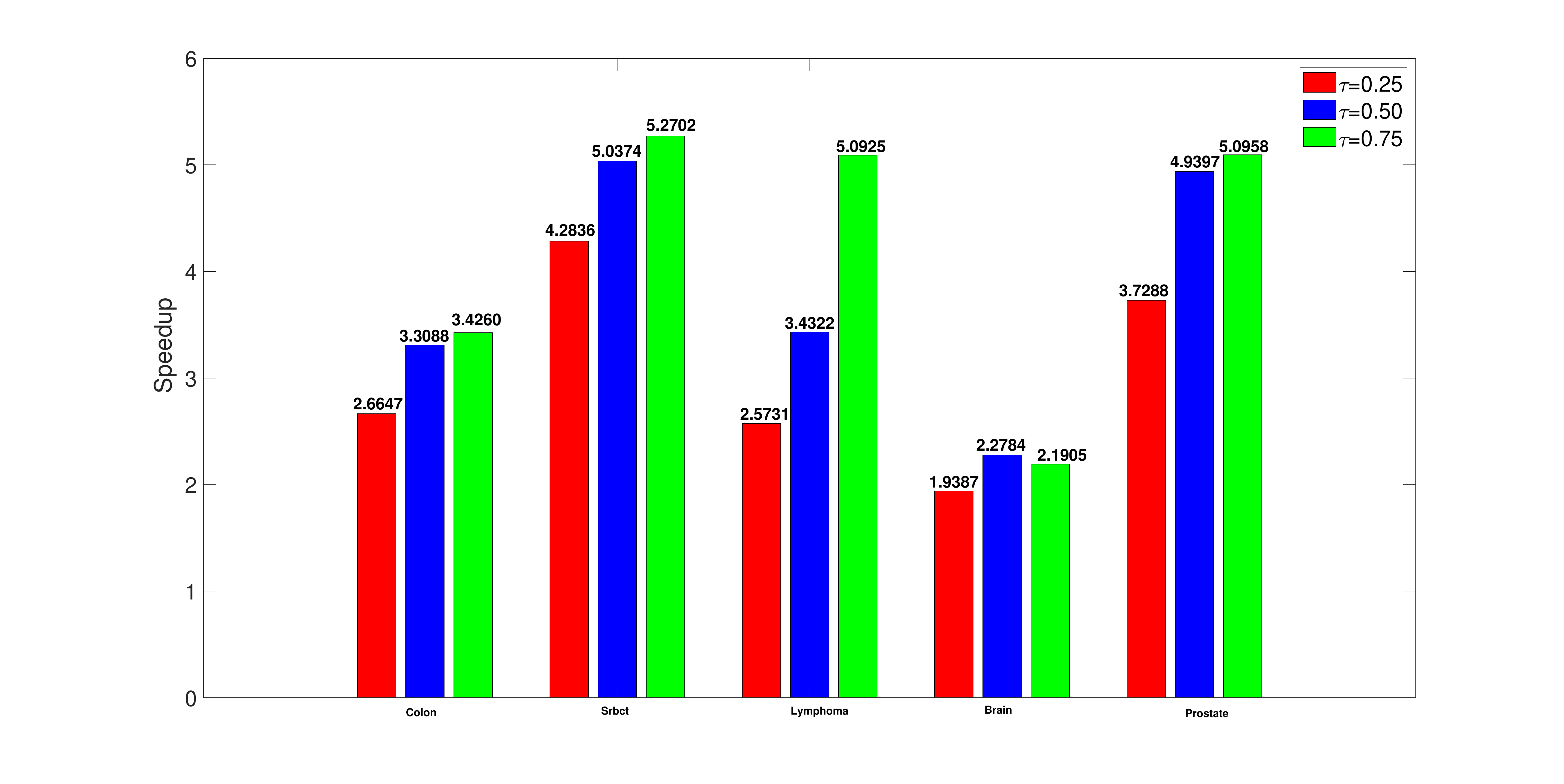}
%\caption{fig2}
\end{minipage}
}%

\subfigure{
\begin{minipage}[t]{1\linewidth}
\centering
\includegraphics[width=3.5in]{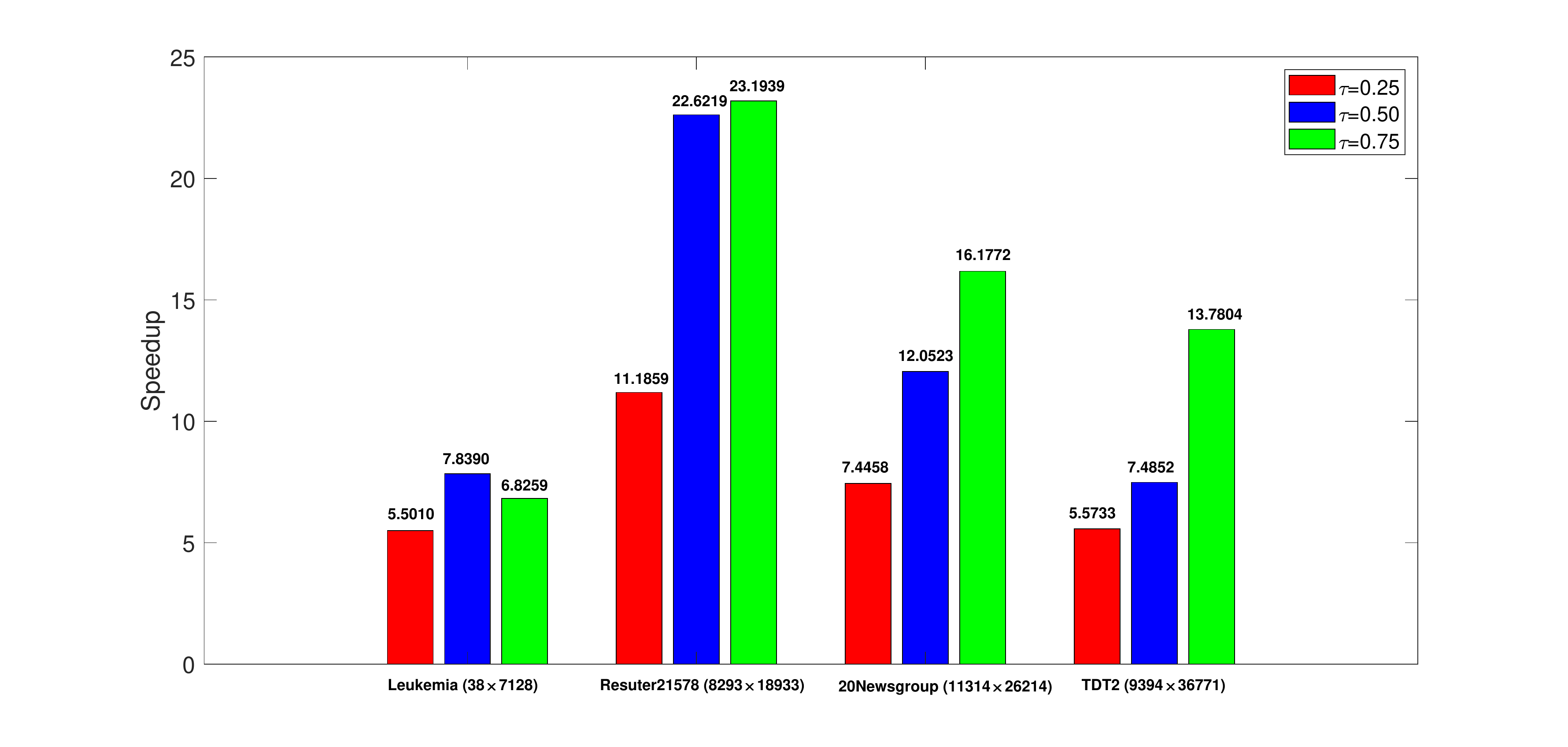}
%\caption{fig2}
\end{minipage}
}%
\caption{The speedup of different datasets under different $\tau$.}
\end{figure}

From results in Fig. 5, Fig. 6 and Fig. 7, we have the following conclusions. (i) Our screening rule eliminates sufficient number of inactive features in all these data sets.   The rejection ratio of these data sets are almost near 1.  (ii) The speedup value differs with the data scale. In these data sets, the speedup of these data sets varies between 2 to 23. (iii) Except the Leukemia and Brain, the speedup value increases with the parameter $\tau$ increasing. According to these results, we conclude that our screening rule is efficient on these high-dimensional data sets.

\section{Conclusion}
In this paper, we introduce the dual circumscribed sphere technique to estimate the dual solution. Based on this technique, we  build up a safe feature screening rule for the $\ell_{1}$-norm  quantile regression, which has a closed-form formula of given data and can be computed with low cost.  In the numerical experiments,  we adopt the proximal ADMM to solve the $\ell_{1}$-norm  quantile regression and illustrate the efficiency of our screening rule, which shows that the screening rule performs well in eliminating inactive features and saving up to 23 times of computational time in some cases.  Actually, our screening rule can be embedded into any algorithm or solver for solving this model.

% or



\ifCLASSOPTIONcompsoc
  \section*{Acknowledgments}
\else
  \section*{Acknowledgment}
\fi This work was supported by the National Science Foundation of China (11671029).

\ifCLASSOPTIONcaptionsoff
  \newpage
\fi

\end{document}